\documentclass{ws-procs9x6}
\usepackage{graphicx, subfigure}

\begin{document}

\title{An epidemiological model of viral infections in a \textit{Varroa}-infested bee colony:
the case of a bee-dependent mite population size}

\author{Sara Bernardi, Ezio Venturino}

\address{Dipartimento di Matematica ``Giuseppe Peano'', \\
via Carlo Alberto 10, \\
Universit\`a di Torino, Italy \\
E-mail: s.bernardi@unito.it, ezio.venturino@unito.it}

\maketitle

{\abstracts{In recent years the spread of the ectoparasitic 
mite \textit{Varroa destructor} has become the most serious threat to worldwide apiculture.
In the model presented here we extend the bee population dynamics with mite viral 
epidemiology examined in an earlier paper
by allowing a bee-dependent mite population size. 
The results of the analysis match field observations well and give a clear explanation of how
\textit{Varroa} affects the epidemiology of certain naturally occurring bee viruses,
causing considerable damages to colonies.
The model allows only four possible stable equilibria, using known field
parameters. 
The first one contains only the thriving healthy bees.
Here the disease is eradicated and also the mites are wiped out. Alternatively,
we find the equilibrium still with no mite population,
but with endemic disease 
among the thriving bee population.
Thirdly, infected bees coexist with the mites in the \textit{Varroa} invasion scenario; in
this situation the disease invades the hive, driving the
healthy bees to extinction and therefore affecting all the bees. 
Coexistence is also possible, with
both populations of bees and mites thriving and with the
disease endemically affecting both species.
The analysis is in line with field observations in natural honey bee colonies. 
Namely,
these diseases are endemic and if the mite population is present, necessarily
the viral infection occurs. Further, in agreement with the fact that the presence 
of \textit{Varroa} increments the viral transmission, the whole bee population may 
become infected when the disease vector is present in the beehive.
Also, a low horizontal transmission rate of the virus
among the honey bees will help in protecting the bee colonies from
\textit{Varroa} infestation and viral epidemics.}}

\section{Introduction}

Current evidence shows that Earth loses between one and ten percent of biodiversity per decade,
a major event of extinction of biological diversity. The causes are mainly
habitat loss, pest invasions, pollution, over-harvesting and diseases, \cite{Intro},
thereby endangering natural ecosystem services that are vital for humanity.

All products of agriculture depend on pollination, that is performed by wild,
free-living organisms such as bees, butterflies, moths and flies. To this end, there is also
the availability of commercially-managed bee species.
But bees represent the most important group of pollinators in most geographical regions,
also for economical reasons, \cite{Intro}. An estimation of the
United Nations Food and Agriculture Organisation (FAO) sets the bee-pollinated
crops at 71 out of the 100 species providing 90\% of food worldwide.
In Europe, there are 264 crop species and
84\% of them are animal-pollinated; further, 4000 vegetable varieties thrive due to
the pollination activities of bees, \cite{Intro}.

As the bee group is the most important pollinator worldwide, this paper focuses on the instability
of bee populations and, in particular, on the the most serious threat to apiculture globally: the
external parasitic mite {\it{Varroa destructor}}, discovered in Southeast Asia in 1904.
This mite is of the size of a pinhead; it feeds on the bees’ circulatory fluid.
It is an invasive species, spreading from one hive to another one.
Today it is present nearly in the whole world, \cite{Intro}.
The serious damage to the bee colonies does not derive from the
parasitic action of the mite but, above all, from its action as vector of many viral diseases.
It increases
the transmission rate of diseases such as acute paralysis (ABPV) and deformed wing viruses (DWV),
that are considered among the main causes of Colony Collapse Disorder (CCD).
This action has been reinforced since about
thirty years ago when the mite has shifted hosts
from {\it{Apis cerana}} to {\it{Apis mellifera}}.
With no control, the infestations are bound to cause the untimely
death of bee colonies within three years.

Mathematical modelling represents a powerful tool for investigating the triangular relationship
between honey bees, {\it{Varroa}} and viral disease. It allows the exploration of the beehive system
without the need of unfeasible or costly field studies. In particular, mathematical models of the
epidemiology of viral diseases on {\it{Varroa}}-infested colonies would allow us to explore the host
population responses to such stressful situations. 

In the next Section, the model is presented and some basic properties
are investigated. In Section \ref{analysis},
equilibria are analysed
for feasibility and stability. A sensitivity analysis on the model parameters is performed
in Section \ref{sens}. A final discussion concludes the study.

\section{The model}

Let $B$ denote healthy bees, $I$ the infected ones, $M$ the healthy mites and $N$ the infected ones
and let the population vector be $X=(B,I,M,N)^T$.

In the model proposed here we essentially extend the epidemiology of a beehive infested by
\textit{Varroa} mites examined earlier in \cite{H}, to account for
a Leslie Gower term in the mite populations, to model the
{\it{Varroa}} population size as a bee population-dependent function,
a step also undertaken in the recent paper \cite{ratti2}.
In fact, the mites essentially carry out their life on the host’s body.
These changes are therefore reflected
in the last two equations, for the mites evolution. The resulting model reads as follows: 
\begin{equation}\label{model}
X'=f(X)=(f_B,f_I,f_M,f_N)^T, \quad f:\mathcal{D}^0 \rightarrow \mathbb{R}_+^4,
\end{equation}
where
\begin{eqnarray*}
\frac {dB}{dt}&=&f_B=b \frac{B}{B+I}-\lambda BN-\gamma BI-m B   
\\ 
\frac {dI}{dt}&=&f_I=b \frac{I}{B+I} +\lambda BN+\gamma BI-(m+\mu)I 
\\
\frac {dM}{dt}&=&f_M=r(M+N) -nM -\frac{p}{h(B+I)}M(M+N)-M(\beta I-\delta N-eB) 
\\
\frac {dN}{dt}&=&f_N=-n N-\frac{p}{h(B+I)} N (N+M)+\beta MI+\delta MN-eNB.  
\end{eqnarray*}
To prevent in the right-hand side of \eqref{model} the vanishing of some terms in the
denominator, we define $\mathcal{D}^0$ as the domain of $f$, explicitly
\begin{equation}
\mathcal{D}^0=\left\{ X=(B,I,M,N) \in \mathbb{R}^4_+ : B+I \neq 0\right\}.
\end{equation}


The first two equations describe respectively the evolution of healthy and infected bee populations.
They are born healthy or infected in proportion to the fraction of healthy or infected bees in the
colony, with constant $b$. The infection process for the larvae indeed occurs mainly through
their meals of contaminated royal jelly.
Healthy bees can contract the virus by infected mites at rate $\lambda$, second terms in
the equations. The next terms model
horizontal transmission of the virus among adult bees,
occuring at rate $\gamma$ via small wounds of
the exoskeleton, e.g. as a result of hair loss, or ingestion of faeces. In the last terms we find
the bees' natural mortality $m$ and the disease-related mortality $\mu$, taken into account for
the infected bee population only. 

The last two equations contain the \textit{Varroa} dynamics, 
partitioned among susceptible and infected states, that vector
the viral disease and once infected remain so for their lifetime.
The mites are always born healthy at rate $r$, as no
viruses are passed vertically from the parents to offsprings.
In the extended model the first three terms describe {\it{Varroa}} growth, natural mortality and
intraspecific competition.
Note that instead in the previously
introduced system, \cite{H}, the mites grow
logistically. Here if the mites reproduction rate $r$ is smaller than
their natural mortality $n$, i.e.
\begin{equation}\label{assump_1}
r<n,
\end{equation}
it is easy to see that the total mite population becomes extinct.
The effect of intraspecific competition among mites is described by the Leslie Gower term at rate
$p$. It accounts for the mite population
dependence on the total number of adult bees, the ``resource for which they
compete'' in the colony, \cite{Ratti}. 
The parameter $h$ expresses the average number of mites per bee.
Note that because the virus does not cause any harm to the
infected \textit{Varroa} population, competition among healthy and
infected mites occurs at the same rate, in other words we find the terms
$M(M+N)$ and $N(M+N)$ with the same ``weight'' in both equations.
If infected were weakened by the disease, we would rather have in the bracket
$M+cN$, with some $c<1$.
Healthy mites can become infected by infected bees at rate $\beta$, the fourth terms in the last
two equations, but they
can also acquire the virus at rate $\delta$
also horizontally from other infected mites.
The last term in the last two equations models the grooming behavior of healthy bees at rate $e$.
We assume that the infected bees do not groom because they are weaker due to the disease
effects. Note that for healthy and infected mites the damage due to grooming is the same, $e$,
as this activity depends only on the bees.


\subsection{Well-posedness and boundedness}

We now address the issue of well-posedness, following basically the path of \cite{FCB}.
The solutions trajectories are shown to be always at a finite distance
from the set 
$\Theta=\left\lbrace (B,I,M,N)  \in \mathbb{R}^4_+ : B+I=0 \right\rbrace $,
where the total bee population vanishes.

\begin{theorem}
Well-posedness and boundedness

\label{teobound}
Let $X_0 \in \mathcal{D}^0$. A solution of \eqref{model}
defined on $[0,+\infty )$ with $X(0)=X_0$ exists uniquely.
Also, for an arbitrary $t>0$, it follows $X(t) \in \mathcal{D}^0$,
and, indicating by $L$ a positive constant,
\begin{eqnarray}\label{bound1}
\frac{b}{\widetilde{m}} \leq \liminf\limits_{t\rightarrow +\infty} (B(t)+I(t))
\leq \limsup\limits_{t\rightarrow +\infty} (B(t)+I(t)) \leq \frac{b}{m}, \quad
\widetilde{m} = m+\mu;\quad \\ \label{bound2}
M(t)+N(t) \leq L , \quad \forall t \geq 0.\qquad 
\end{eqnarray}
\end{theorem}

\begin{proof}
We follow the arguments of \cite{FCB,H} for the first part of the proof.
Global Lipschitz continuity of the right-hand side of the system holds 
in $\mathcal{D}^0$, thereby implying existence and uniqueness of the solution of
system \eqref{model} for every trajectory at a finite
distance of this boundary. 
Now \eqref{bound1} and \eqref{bound2} hold for all the
trajectories starting at any point with $B+I \neq 0$.
The boundedness
of the variables entail consequently that all trajectories exist at all times in the future
and are bounded away from the set $\Theta$.

Adding the two equations for the bees in \eqref{model}, at any point in $\mathcal{D}^0$
we have
\begin{equation*}
B'+I'=b-mB-(m+\mu)I \geq b-\widetilde{m}(B+I), \quad \widetilde{m}=m+\mu.
\end{equation*}
Integration of this differential inequality between $X(0)=X_0$ and $X(t)$,
points belonging to a trajectory such that
$X(\tau) \in \mathcal{D}^0$ for all $\tau \in [0,t]$, we obtain the following lower bound, with
positive right hand side at any $t>0$,
\begin{equation}\label{bound1a}
B(t)+I(t) \geq \frac{b}{\widetilde{m}}(1-e^{-\widetilde{m}t})+(B(0)+I(0))e^{-\widetilde{m}t}.
\end{equation} 

Similarly, but bounding from above, we find $B'+I' \leq b-m(B+I)$,
and therefore
\begin{equation}\label{bound1b}
B(t)+I(t) \leq \frac{b}{m}(1-e^{-mt})+(B(0)+I(0))e^{-mt}.
\end{equation}
The inequalities
in \eqref{bound1} for any portion of the trajectory belonging to $\mathcal{D}^0$ follow
from \eqref{bound1a} and \eqref{bound1b}.

We now turn to the \textit{Varroa} populations $M$ and $N$, and their
the whole mite population $V=M+N$, summing their corresponding
equations of \eqref{model} and using \eqref{bound1}, for any point within $\mathcal{D}^0$
we have the upper bound
%
\begin{eqnarray}\label{magg}
V'+\alpha V =(r-n+\alpha)V-\frac{p}{h(B+I)}V^2  -eBV \qquad \\ \nonumber
\leq \left[ r-n+\alpha-\frac{p}{h(B+I)}V\right] V  \quad
\leq \left[ r-n+\alpha-\frac{pm}{hb}V \right] V = \phi(V) \leq \phi(\overline{V}),
\end{eqnarray}
where $\phi(\overline{V})$ represents the maximum value of the parabola $\phi(V)$,
namely $\overline{V}=hb(r-n+\alpha)(2pm)^{-1}$.
%
%

A consequence of \eqref{magg} is the differential inequality
$V' \leq \phi(\overline{V}) - \alpha V$
and thus integrating it we obtain the following inequality, for all $t \geq 0$,
proving \eqref{bound2},
\begin{equation}\label{mite_bound}
V(t) \leq \frac{\phi(\overline{V})}{\alpha} \left(1-e^{-\alpha t} \right) + V(0)e^{-\alpha t}
\leq \max \left\lbrace  V(0), \frac{\phi(\overline{V})}{\alpha} \right\rbrace.
\end{equation}
From (\ref{bound1a}), (\ref{bound1b}) and (\ref{mite_bound}) the boundedness of all populations is thus ensured,
\cite{FCB}.
Thus all trajectories originating in $\mathcal{D}^0$
remain in $\mathcal{D}^0$ for all $t>0$. \quad
\end{proof}

Let $\mathcal{D}^1$ be the largest subset of $\mathcal{D}^0$ satisfying
the inequalities of Theorem \eqref{teobound},
\begin{equation*}
\mathcal{D}^1 \doteq \left\lbrace (B,I,M,N) \in \mathbb{R}^4_+ : \frac{b}{\widetilde{m}}
\leq B+I \leq \frac{b}{m}, 0 \leq M+N \leq L \right\rbrace.
\end{equation*}
As a consequence of Theorem \ref{teobound},
it is a compact set, positively invariant for all the system's trajectories.
All the system's equilibria belong to $\mathcal{D}^1$, as it is shown later.
Thus the equilibria analysis is sufficient to explain the whole system's behavior.

\section{Equilibria}\label{analysis}

The equilibria $E_k=(B_k, I_k, M_k, N_k)$ of \eqref{model} are
$$
E_1=\left(\frac{b}{m},0,0,0\right),\quad E_2=\left( 0,\frac{b}{m+\mu},0,0\right),
$$
which are always feasible;
the points
$$
E_3=\left(\frac{b}{m},0,\left( r-\frac{eb}{m}\right)\frac{bh}{rm},0\right), \quad
E_4=\left(\frac{\mu^2-b \gamma +m \mu}{\mu \gamma}, \frac{b \gamma - m \mu}{\mu \gamma},0,0\right),
$$
the former feasible for
\begin{equation}\label{E3feas}
r \geq \frac{eb}{m},
\end{equation}
the latter feasible whenever the following condition is satisfied
\begin{equation}\label{E4feas}
0<b \gamma - m \mu <\mu^2.
\end{equation}

Then, there is the coexistence equilibrium $E_*=(B_*, I_*, M_*, N_*)$. However, it is not
analytically tractable and will therefore be investigated numerically, with the help of simulations.

Finally, we find the equilibrium $E_5=(0,I_5, M_5, N_5)$ with no healthy bees. It is obtained as
intersection of two conic sections that lie in the fist quadrant of the $M$ - $N$ phase plane. We
discuss it in the next subsection.

\subsection{The healthy-bees-free equilibrium}

From the second equation we find $I=b(m+\mu)^{-1}$.
Substituting the value of $I$, the last two equations can be rewritten as 
 
\begin{equation}
\label{psi}
\psi: (r-n)M+rN-\frac{p(m+\mu)M(M+N)}{hb}-\frac{\beta b M}{m+\mu}- \delta MN=0,
\end{equation}
and
\begin{equation}
\label{eta}
\eta: -nN-\frac{p (m+\mu) N(N+M)}{hb}+\frac{\beta b M}{m+\mu}+\delta MN=0.
\end{equation}

Thus, the equilibrium follows from determining the intersection of
$\psi$ and $\eta$ in the fist quadrant of the $M - N$ phase plane.

We begin by analyzing the curve $\psi$.

Solving \eqref{psi} for the variable $N$, we find 
\begin{equation}
\label{psi_expl} 
\psi: N=M \left[ \frac{p(m+\mu)^2 M +\beta h b^2 + (n-r)hb(m+\mu)}{(m+\mu)(hbr-p(m+\mu)M-\delta h b M)} \right]=M \tilde{\psi}(M)
\end{equation}
The conic $\psi$ crosses the vertical axis at the origin and the horizontal one at the absissa 
\begin{equation}
M_0=\frac{hb[(r-n)(m+\mu)-\beta b]}{p(m+\mu)^2}
\end{equation}
Further, there is the vertical asymptote
\begin{equation}
M_{\infty}=\frac{hbr}{p(m+\mu)+\delta h b}>0.
\end{equation}
From \eqref{fig:psi}, it is easy to assess the signs of the numerator and the denominator,
respectively
$\mathcal{N}>0$ for $M>M_0$
and $\mathcal{D}>0$ for $M<M_{\infty}$.
Depending on the sign of $\mathcal{N}$ and $\mathcal{D}$, $\psi_1$ thus
shows three different shapes,
Figure \ref{fig:psi}.

\begin{figure}[ht!]
\centering \textbf{Hyperbola $\psi$} \par \medskip
{\includegraphics[scale=.18]{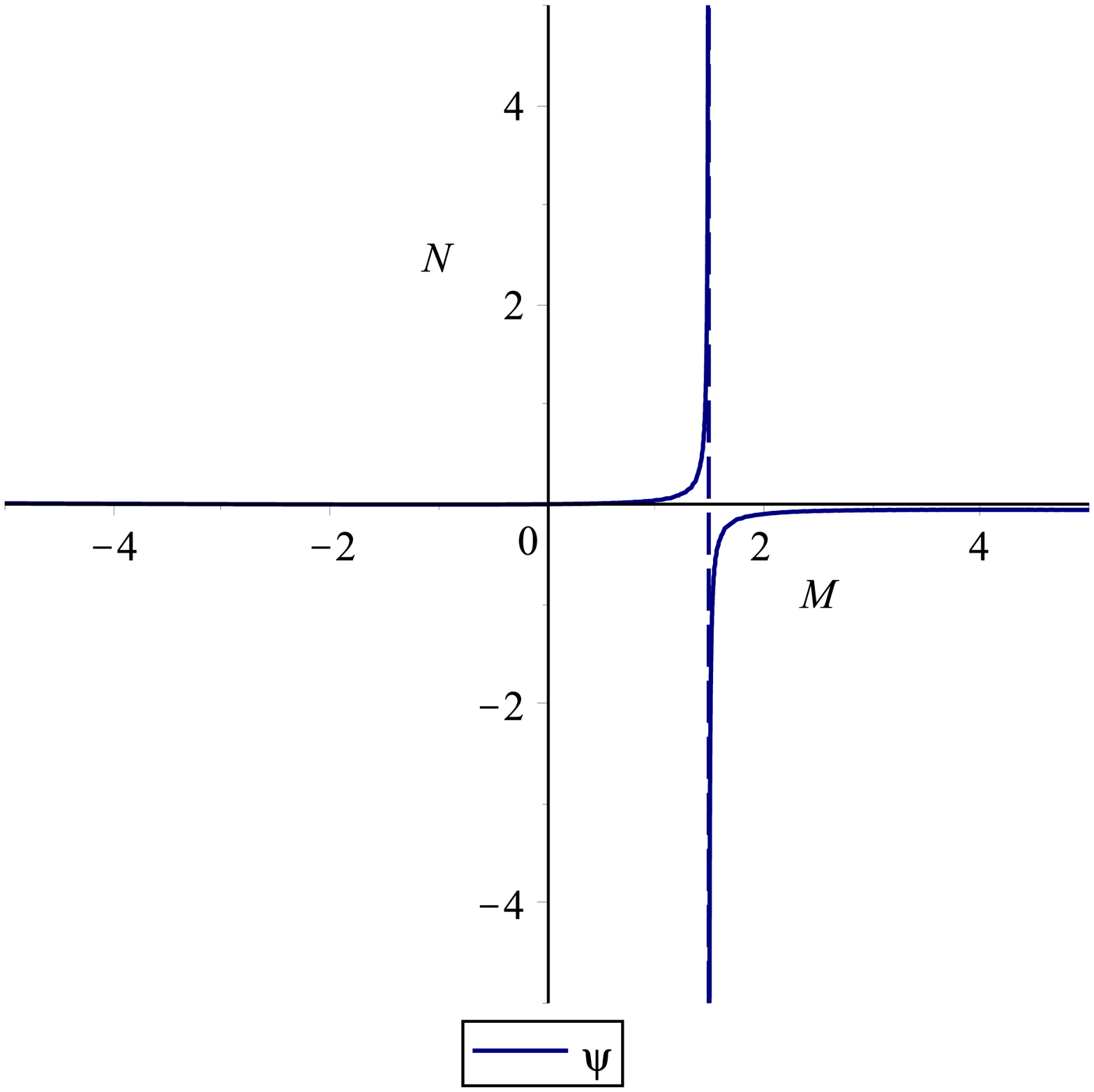}}\quad
{\includegraphics[scale=.18]{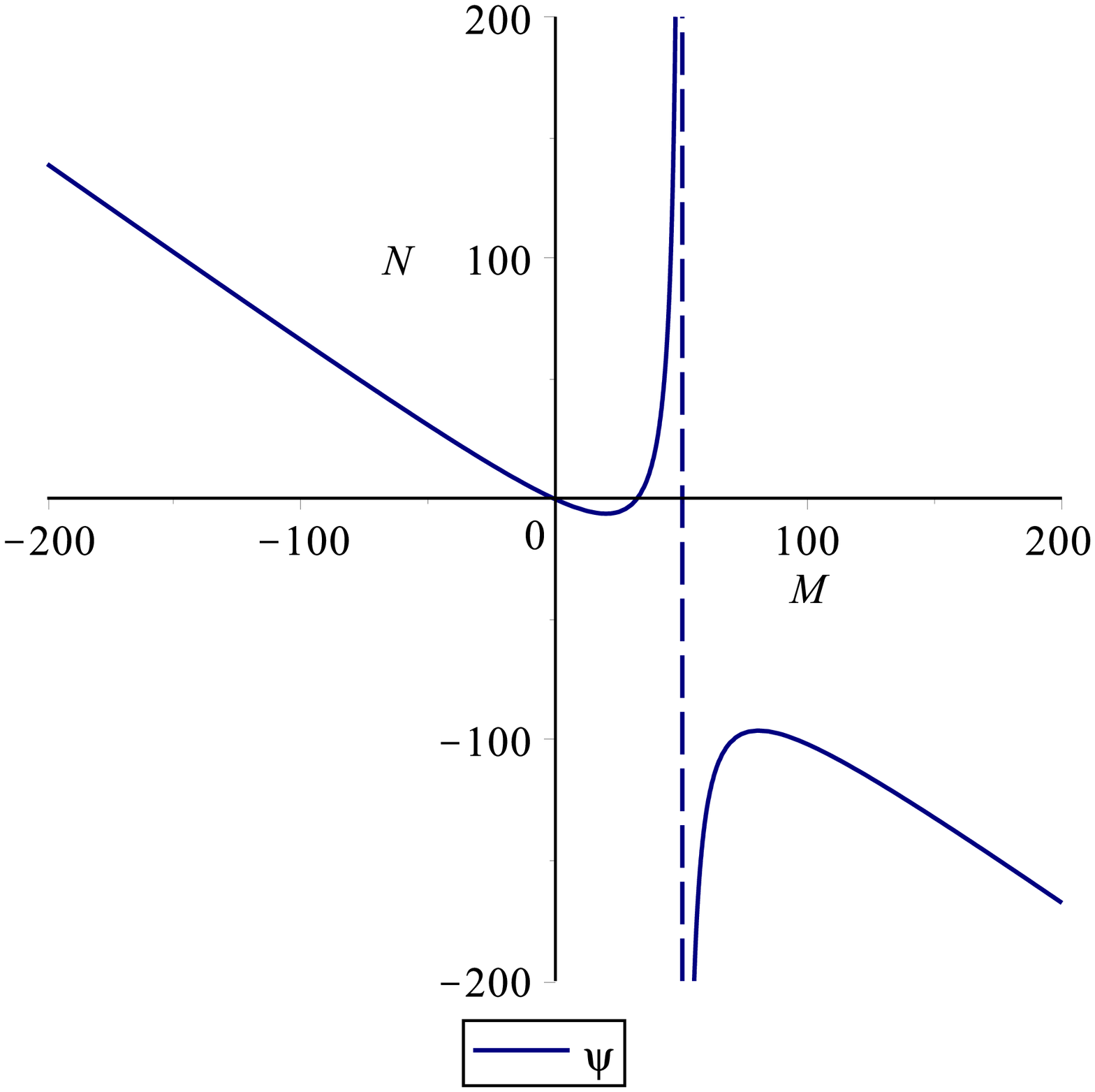}}\quad
{\includegraphics[scale=.18]{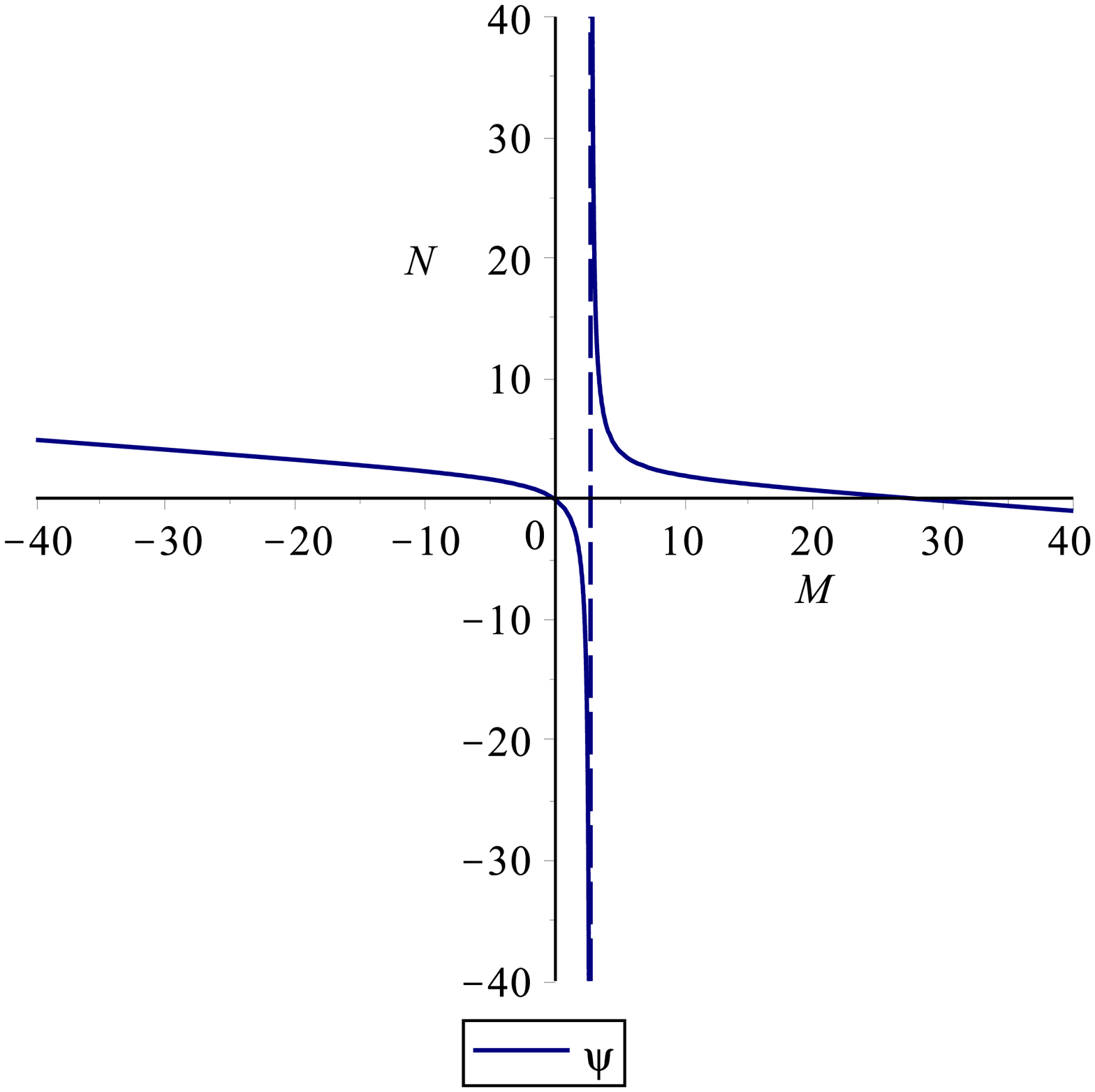}}
\caption{
Left: Case 1. $M_0<0<M_{\infty}$, for the parameters values $r=0.06$, $e=0.001$, $\lambda=0.03$,
$b=2150$, $\beta=0.0002$, $m=0.023$, $\mu=8$, $\delta=0.04$, $p=0.05$, $h=0.9$, $n=0.007$.
Center: Case 2. $0<M_0<M_{\infty}$, for the parameters values $r= 0.2$, $e= 0.001$, $\lambda=0.03$,
$b= 150$, $\beta= 0.0002$, $m=0.023$, $\mu=8$, $\delta=0.001$, $p=0.05$, $h= 0.9$, $n= 0.1$.
Right: Case 3. $0<M_{\infty}<M_0$, for the parameters values $r=0.06$, $e=0.001$, $\lambda=0.03$,
$b= 250$, $\beta= 0.0002$, $ m=0.04$, $\mu= 8$, $\delta= 0.02$.
}
\label{fig:psi}
\end{figure}

We now turn to the second conic section, $\eta$.

Rearranging \eqref{eta} for $M$, we get the explicit form
$$
M=N \frac{[hbn+p(m+\mu)N](m+\mu)}{\beta b^2 h +[\delta h b(m+\mu)-p(m+\mu)^2]N }=N \tilde{\eta}(N)
$$
Proceeding as before, we determine the intersections of $\eta$ with the axes. 
The conic $\eta$ goes through the origin and crosses the $N$ axis at
$$
N_0=-\frac{hbn}{p(m+\mu)}<0.
$$
Again, we get one vertical asymptote
$$
N_{\infty}=\frac{\beta b^2 h}{(m+\mu)[p(m+\mu)-\delta h b]}.
$$
Now from \eqref{eta}, the sign of the numerator is positive, $\mathcal{N}>0$, for $N>N_0$,
while for the denominator $\mathcal{D}$ two cases arise. Namely, if $\delta h b< p(m+\mu)$,
i.e.\ $N_{\infty}>0$, we get $\mathcal{D}>0$ for $N<N_{\infty}$. Otherwise,
if $\delta h b> p(m+\mu)$, i.e.\ $N_{\infty}<0$, to have $\mathcal{D}>0$ the opposite
condition $N>N_{\infty}$ must hold.

Figure \ref{fig:eta} sums up the three possible shapes for $\eta$.

\begin{figure}[h!]
\centering \textbf{Hyperbola $\eta$} \par \medskip
\includegraphics[scale=.18]{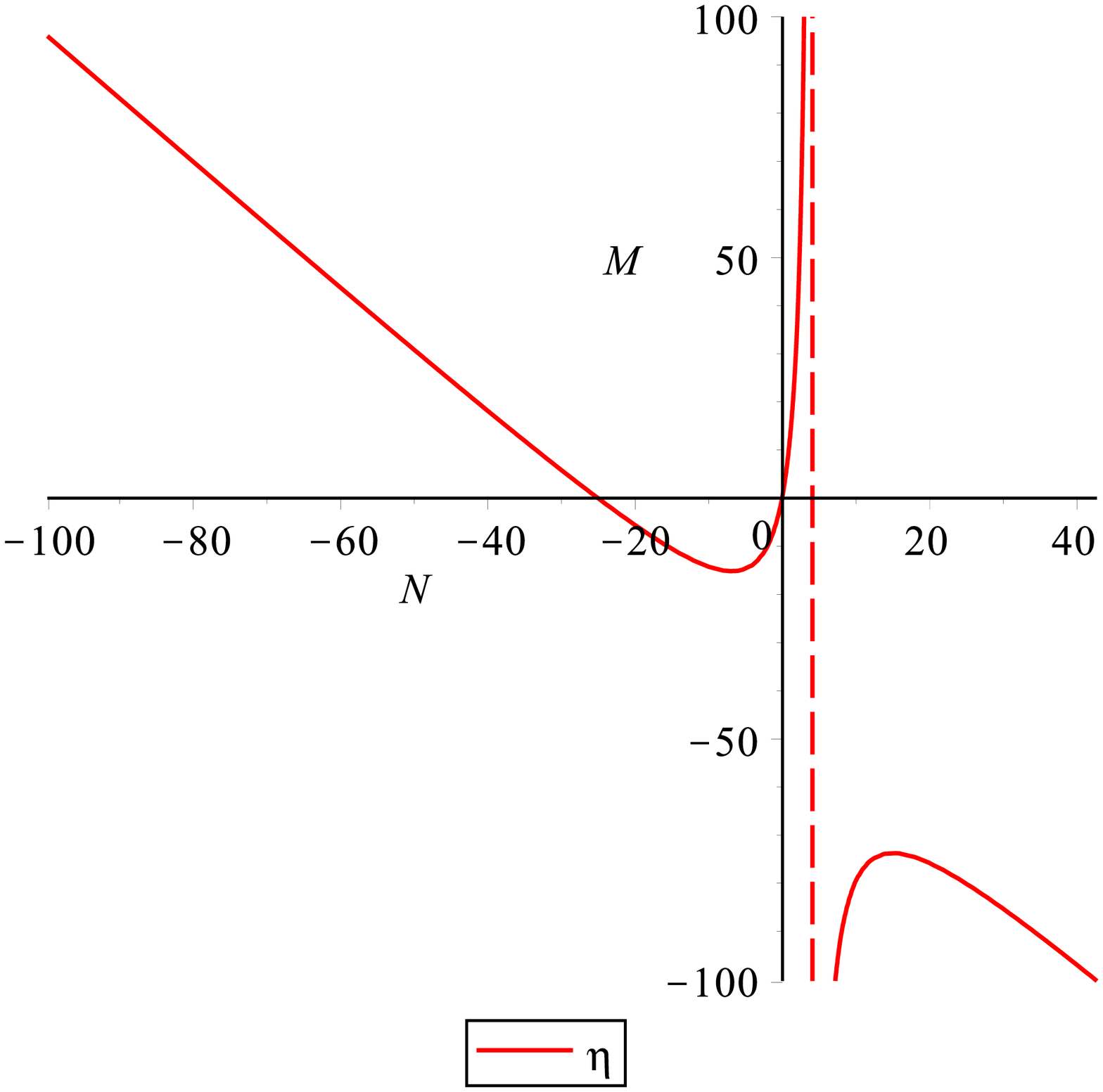}\quad
\includegraphics[scale=.18]{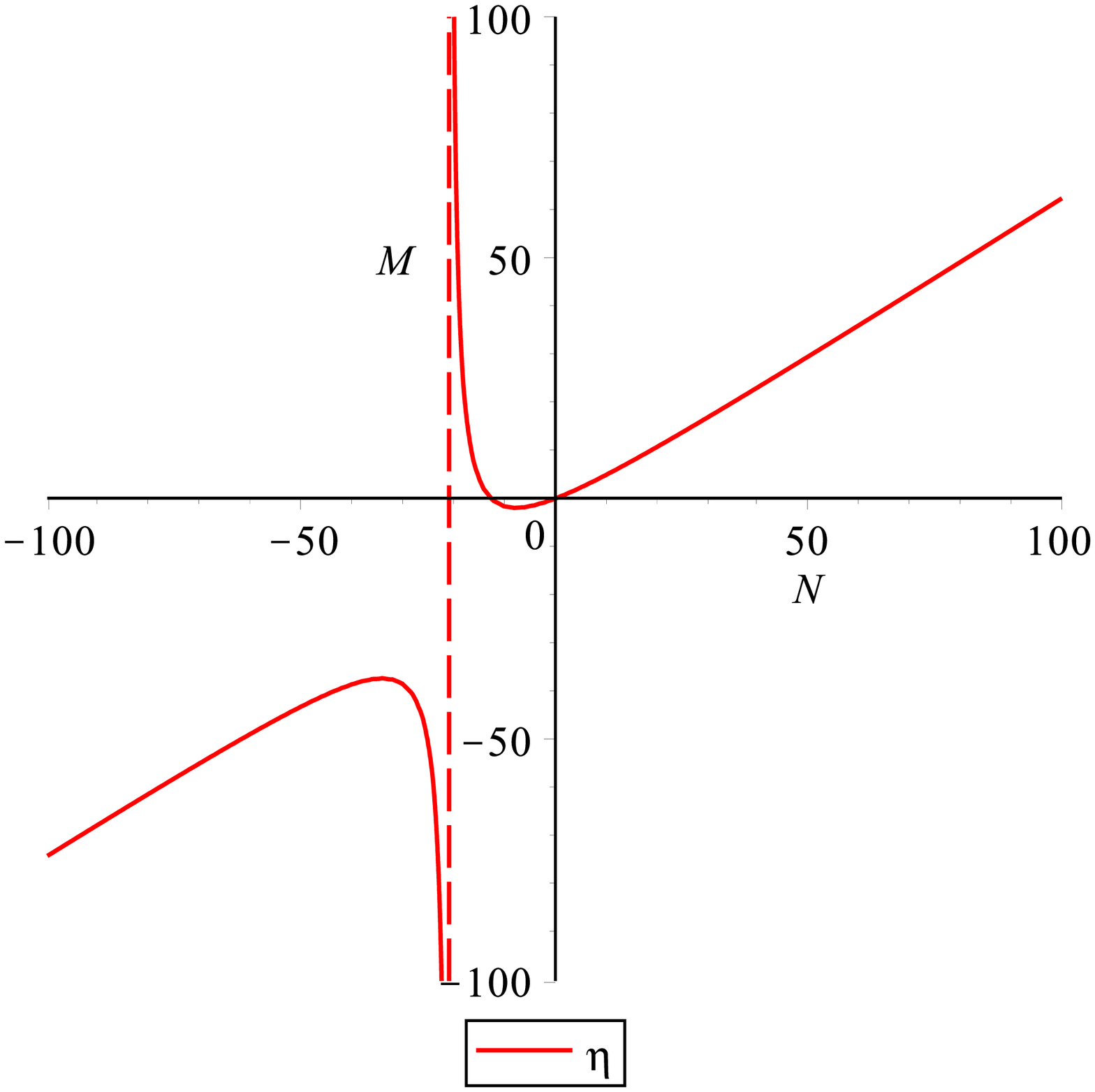}\quad
\includegraphics[scale=.18]{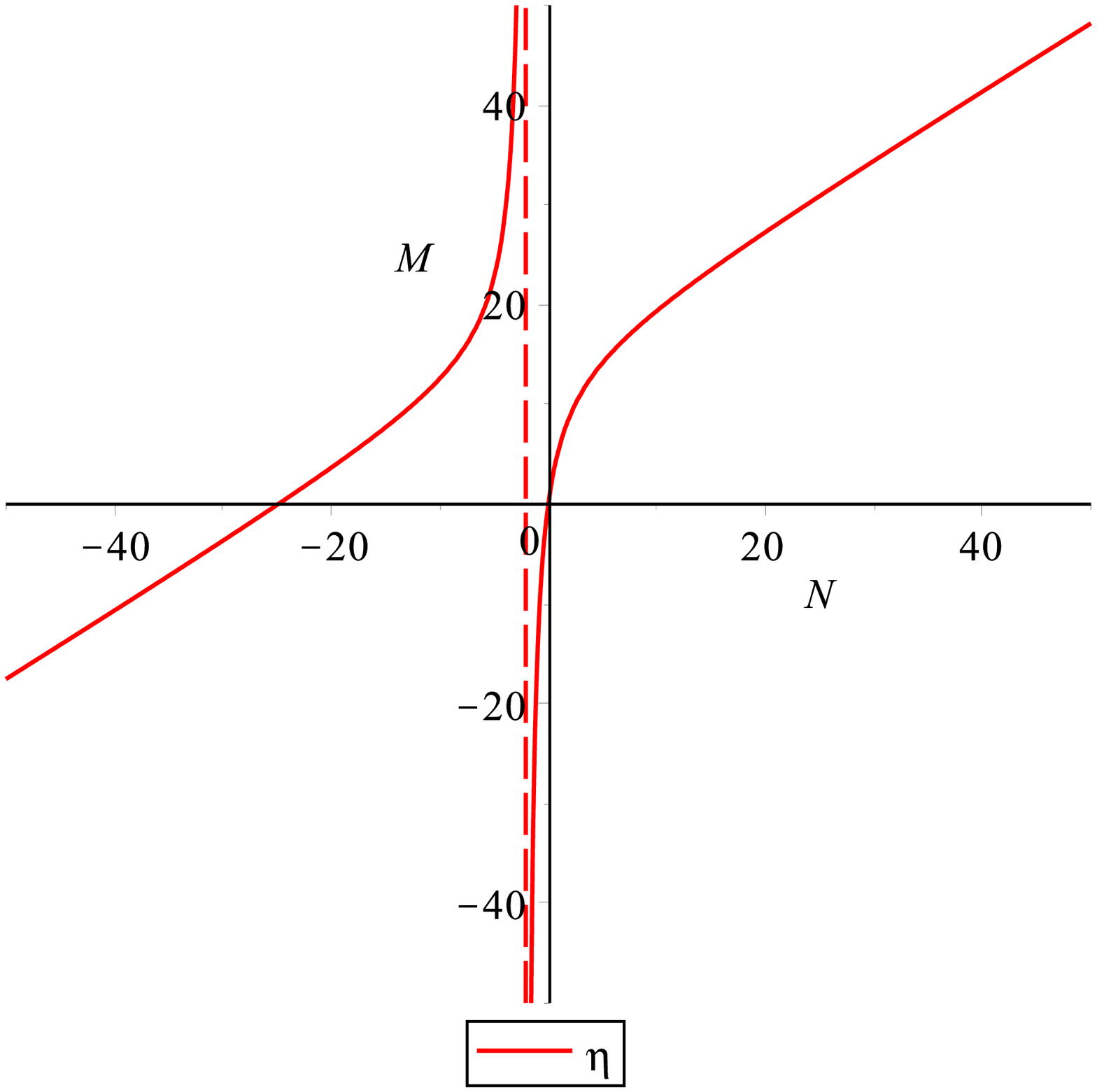}
\caption{
Left: Case A. $N_0<0<N_{\infty}$, for the parameters values $r= 0.2$, $e=0.001$, $\lambda=0.03$,
$b= 100$, $\beta= 0.001$, $ m= 0.023$, $\mu = 8$, $\delta= 0.001$, $p= 0.05$, $h = 1$, $n= 0.1$.
Center: Case B. $N_{\infty}<N_0<0$, for the parameters values $r=0.2$, $e= 0.001$, $\lambda=0.03$,
$b= 100$, $\beta = 0.01$, $m= 0.023$, $\mu= 8$, $\delta= 0.01$, $p= 0.05$, $h= 1$, $n= 0.05$.
Right: Case C. $N_0<N_{\infty}<0$, for the parameters values $r= 0.2$, $e= 0.001$, $\lambda = 0.03$,
$b = 100$, $\beta= 0.001$, $m=0.023$, $\mu = 8$, $\delta = 0.01$, $p = 0.05$, $h = 1$, $n = 0.1$. 
}
\label{fig:eta}
\end{figure}

Finally, by graphically plotting both the hyperbolae in the $M$-$N$ plane,
they meet at the origin and they further intersect at another point located in the first quadrant.
This immediately provides the unconditional existence and feasibility
of the equilibrium point $E_5$.

\subsection{Stability}

The Jacobian $J$ of (\ref{model}) is the following matrix
$$
\left(
\begin{array}{cccc}
J_{11} & -\dfrac{bB}{(B+I)^2} - \gamma B & 0 & - \lambda B \\
J_{21} & \dfrac b{B+I} - \dfrac {bI}{(B+I)^2} + \gamma B -(m+\mu) & 0 & \lambda B \\
J_{31} &  \dfrac {pM(M+N)}{h(B+I)^2}  -\beta M & J_{33} & r-\dfrac{pM}{h(B+I)} -\delta M \\
J_{41} &  \dfrac {pN(M+N)}{h(B+I)^2}  +\beta M & J_{43} &
-n- \frac {p(M+2N)}{h(B+I)}  +\delta M - eB \\
\end{array}
\right)
$$
with
\begin{eqnarray*}
J_{11}=\frac b{B+I} - \frac {bB}{(B+I)^2} - \lambda N - \gamma I -m,\quad
J_{21}=- \frac {bI}{(B+I)^2} + \lambda N + \gamma I,\\
J_{31}= \frac {pM(M+N)}{h(B+I)^2}  -eM,\quad
J_{33}=r-n-\frac {p(2M+N)}{h(B+I)} -\beta I - \delta N -e B,\\
J_{41}= \frac {pN(M+N)}{h(B+I)^2}  -eN,\quad
J_{43}=- \frac {pN}{h(B+I)}  +\beta I + \delta N.
\end{eqnarray*}

At $E_1$ two eigenvalues of the the Jacobian evaluated at this point, $J(E_1)$, are negative,
$-m$ and $-n-ebm^{-1}$. The remaining two provide the stability conditions
\begin{equation}\label{E1_stab}
\gamma b< \mu m, \quad
mr< eb+ nm.
\end{equation}

At equilibrium $E_2$ the Jacobian $J(E_2)$ has two explicit eigenvalues, $-(m+\mu)<0$
and $\mu -\gamma b (m+\mu)^{-1}$
while the Routh-Hurwitz conditions on the remaining minor
show that if (\ref{assump_1}) is not satisfied, the equilibrium is unstable,
because in such case $\det (J(E_2))<0$,
$$
-{\textrm{tr}} (J(E_2))=2n+\frac {b\beta}{m+\mu}-r, \quad
\det (J(E_2))=(n-r)\left( n+\frac {b\beta}{m+\mu}\right).
$$

The Jacobian $J(E_3)$ once again gives two explicit eigenvalues,
$-m<0$ and another one providing the first stability condition
$ r< n + 2pM_3(hB_3)^{-1} + eB_3$, i.e. explicitly
\begin{equation}\label{E3_stab1}
2eb < m (r-n),
\end{equation}
and the Routh-Hurwitz criterion on the remaining minor $\widetilde J_3$ gives the further
conditions
\begin{eqnarray}\label{E3_stab2}
-{\textrm{tr}} (\widetilde J_3)=-\frac b{B_3} - \gamma B_3 +m+\mu+n+\frac {pM_3}{hB_3}
- \delta M_3+eB_3>0, \qquad \\
\label{E3_stab3}
\det (\widetilde J_3)=\left( \frac b{B_3} +\gamma B_3 -m-\mu \right)
\left( \delta M_3 -eB_3 -n -\frac {pM_3}{hB_3}\right) - \beta \lambda B_3 M_3 >0. \qquad
\end{eqnarray}
Note that if (\ref{assump_1}) holds, $E_3$ is unstable, because in this case (\ref{E3_stab1})
cannot be satisfied.

At the point $E_4$ the characteristic equation factorizes into the product of two
quadratic equations, for which the Routh-Hurwitz conditions provide the following pairs of
inequalities, to be satisfied for stability
\begin{eqnarray}\label{E4_stab1}
2m+\mu + \gamma I_4 > \gamma B_4 + \frac {b}{B_4+I_4}, \\ \nonumber
B_4 I_4 \left[ \gamma^2 - \frac {b^2}{(B_4+I_4)^4} \right]
+ \left( \frac b{B_4+I_4} -m \right)^2 
\qquad \\ 
\nonumber
+ \left( \frac b{B_4+I_4} - m \right) \left[ \gamma B_4 - 
\frac {bI_4}{(B_4+I_4)^2} - \mu - \frac {bB_4} {(B_4+I_4)^2}  \right] \qquad \\ \label{E4_stab3}
> \left( \gamma I_4 + \frac {bB_4}{(B_4+I_4)^2} \right)
\left[ \gamma B_4 - \mu - \frac {bI_4}{(B_4+I_4)^2}\right], \qquad
\end{eqnarray}
and
\begin{eqnarray}\label{E4_stab2}
2n+2eB_4 + \beta I_4 > r, \quad
(n + \beta I_4 + e B_4) (n+eB_4)> r[ (n+eB_4) + \beta  I_4]. \qquad
\end{eqnarray}

Two eigenvalues are also explicitly found at $E_5$, $-m-\mu<0$ and the other one giving
the first stability condition
\begin{eqnarray}\label{E5_stab1}
\frac b{I_5} < \gamma I_5 + m + \lambda N_5.
\end{eqnarray}
The Routh-Hurwitz conditions for stability on the remaining quadratic become then
\begin{eqnarray}\label{E5_stab2}
2n+ \frac {3p(M_5+N_5)}{hI_5} + \beta I_5 + \delta N_5 > r + \delta M_5,\\ \nonumber
\left( r-n - \frac {p(2M_5 +N_5)}{hI_5} -\beta I_5 - \delta N_5 \right)
\left( \delta M_5 -n - \frac {p(M_5 +2N_5)}{hI_5} \right)\\ \nonumber
+ \left( \beta I_5 - \frac {p N_5}{hI_5} + \delta N_5 \right)
\left( \frac {p M_5}{hI_5} + \delta M_5 - r \right) > 0.
\end{eqnarray}

\section{Results}

After describing the set of parameter values used, the chosen initial conditions
and the field data available, we perform the sensitivity
analysis. The simulations have been performed using the Matlab built-in ordinary differential
equations solver ode45.

\subsection{Model parameters from the literature}

The model parameters that are known from the literature,
\cite{izslt}, \cite{joyce}, \cite{maci} and \cite{api2}, \S 8.2.3.5, 
are fixed in the simulations, while the remaining ones are changed
over a suitable range.
In the sensitivity analysis however we will vary also the known ones,
to simulate possible environmental variations, due perhaps to climatic changes
or other external disruptions.

We set the time unit to be the day.
The worker honey bees birth rate is $b=1500$,
their natural mortality rate instead is $m=0.023$,
which implies a life expectancy of $43.5$ days in the adult stage, \cite{izslt}.
There are no precise values for the grooming behavior in the literature.
A possible range of $e$ is presumed to be in the interval $[10^{-6},10^{-5}]$, \cite{joyce}.

The \textit{Varroa} population reproduces
exponentially fast, doubling every month during the spring and summer.
We then take $r \approx 30^{-1} \ln 2 $, i.e.\ $r=0.02$, \cite{api2}, p. 225.
In the same season,
the mite natural mortality rate in the phoretic phase, i.e.\ when attached to an adult
bee, is presumed
to have the value $n=0.007$, \cite{maci}

In Table 1 we list all the reference values for the numerical experiments.
The other model parameters are freely chosen, with hypothetical values.

\begin{table}[h!]\label{table3}
{{\textbf{Table 1}}. For these model parameters the values are obtained from the literature.
Bee and mite populations are measured in pure numbers.}
\begin{center}
\footnotesize
\begin{tabular}{|c|c|c|c|c|}
\hline
\textbf{Parameter}     &\textbf{Interpretation}        &\textbf{Value}  &\textbf{Unit}  &\textbf{Source}  \\ \hline
\textbf{$b$}            &Bee daily birth rate          &$1500$       &day$^{-1}$  &\cite{izslt} \\   \hline 
\textbf{$e$}            &Healthy bee grooming rate    &$10^{-6} - 10^{-5}$ &day$^{-1}$   &\cite{joyce}      \\   \hline 
\textbf{$m$}        &Bee natural mortality rate    &$0.023$ &day$^{-1}$ &\cite{izslt}    \\   \hline
\textbf{$r$}         &\textit{Varroa} growth rate    &$0.02$ &day$^{-1}$ &\cite{api2}    \\   \hline
\textbf{$n$}       &\textit{Varroa} natural mortality &$0.007$ &day$^{-1}$ &\cite{maci} \\ 
&rate in the phoretic phase & & & \\ \hline
\end{tabular}
\end{center}
\end{table}

\subsection{Setting of initial conditions and free parameters}

The colony conditions at the beginning of the spring
come from field data:
all the bees are healthy. Indeed the infected ones do not survive the winter, because
they have a lower life expectancy. In addition,
the colony treatments with acaricide are usually performed
in the late autumn, to allow the mite eradication.
We can safely assume then the mite population is around $10$ units at the start of the spring.


\subsection{Use of field data}

Before proceeding with the sensitivity analysis, we check that the feasibility and stability
conditions of the equilibrium points are satisfied by
the known parameter values, see Table 2.

In particular, we remark that the disease-free equilibrium point $E_3$
never occurs.
Namely, the feasibility condition \eqref{E3feas} does not hold in field conditions.

This result highlights the close connection between \textit{Varroa} and viruses.
Further, it matches well with beekeeper observations: the bees viral infection is bound to occur
whenever the mite population is present. The same conclusion has been obtained
in our previous study, \cite{H}.

\begin{table}[h!]
\label{table2}
{{\textbf{Table 2}}. Summary of the equilibria: feasibility and stability conditions}
\begin{center}
\begin{tabular}{|c|c|c|c|}
\hline
\textbf{Equilibrium}     &\textbf{Feasibility}        &\textbf{Stability}  &
\textbf{In field} \\ &&& \textbf{conditions} \\ \hline
\textbf{$E_1$}            &always          &\eqref{E1_stab} & allowed    \\   \hline 
\textbf{$E_2$}            &always          &unstable & unstable     \\   \hline 
\textbf{$E_3$}         &\eqref{E3feas}      &\eqref{E3_stab1}, \eqref{E3_stab2}, \eqref{E3_stab3} & infeasible   \\   \hline
\textbf{$E_4$}         &\eqref{E4feas}    &\eqref{E4_stab1}, \eqref{E4_stab3}, \eqref{E4_stab2}  &  allowed \\   \hline
\textbf{$E_5$}         & always
&\eqref{E5_stab1}, \eqref{E5_stab2} &  allowed \\   \hline
\textbf{$E_*$}       & numerical & numerical & allowed \\ 
       & simulations & simulations & \\ \hline
\end{tabular}
\end{center}
\end{table}


\subsection{Sensitivity analysis}\label{sens}
In this section we investigate the behavior of the system responses when the parameters
change their values within an appropriate range.
We compute the various populations equilibrium values as function
of a pair of parameters at a time, thus obtaining surfaces in all the possible pairs of
parameter spaces.
In order to do this, for each pair of parameters we
combine the respective ranges to build a equispaced grid of values and then we plot the four
surfaces resulting from the values assumed by the populations at each point of the grid.
Since our model contains $12$ parameters, there are $66$ possible cases.
Of these, we present the most interesting results starting from
the best situation for the hive, i.e. the cases in which the
mite-and-disease-free equilibrium $E_1$ is stably attained.
They arise for the following parameter pairs: $h-\gamma$, $b-\mu$ and $m-\mu$.

In Figure \ref{surf_E1} (left)
the system shows different transcritical bifurcations, as the parameter $\gamma$ increases. 
Specifically, for smaller values of the horizontal trasmission rate $\gamma$ the system settles
to $E_1$, the healthy beehive scenario, then the infected bee population appears in the system,
i.e.\ we find the epidemics among the bees, $E_4$, and finally $E_5$,
where the healthy bees population disappears and the mites invade the hive.
An increase in $h$ instead positively affects the mite populations, as expected.

From Figure \ref{surf_E1} (right), increasing the bees disease-related mortality rate $\mu$
drives the system into a safer scenario, all bees are healthy and the colony becomes mite-free.
A higher bees birth rate seems to favor the infected mites and bees populations, when the former
appear in the ecosystem and also the healthy bees, but this is true only for low values of $b$.
Instead the susceptible mites decrease for an increasing $b$.
Here transcritical bifurcations relate
all the equilibrium points.
The results confirm what already remarked in \cite{H}.  
Namely, the equilibrium $E_1$ is reached for very small values of the
transmission rate $\gamma$ combined with a high enough bees disease-related mortality $\mu$. 
The behavior of the sensitivity surface in the $m-\mu$ case is very similar.

In the parameter spaces
$b-\lambda$ and $b-\gamma$, not shown, a low $\gamma$ or $\lambda$ favors
the healthy bees and both healthy and infected mites,
while $b$ fosters all the system's populations.
For $b-\beta$ and $b-\delta$ the behavior is similar in terms of $b$, but decreasing
the other parameter depresses instead the infected mites.


\begin{figure}[h!]
\centering 
\includegraphics[scale=.3]{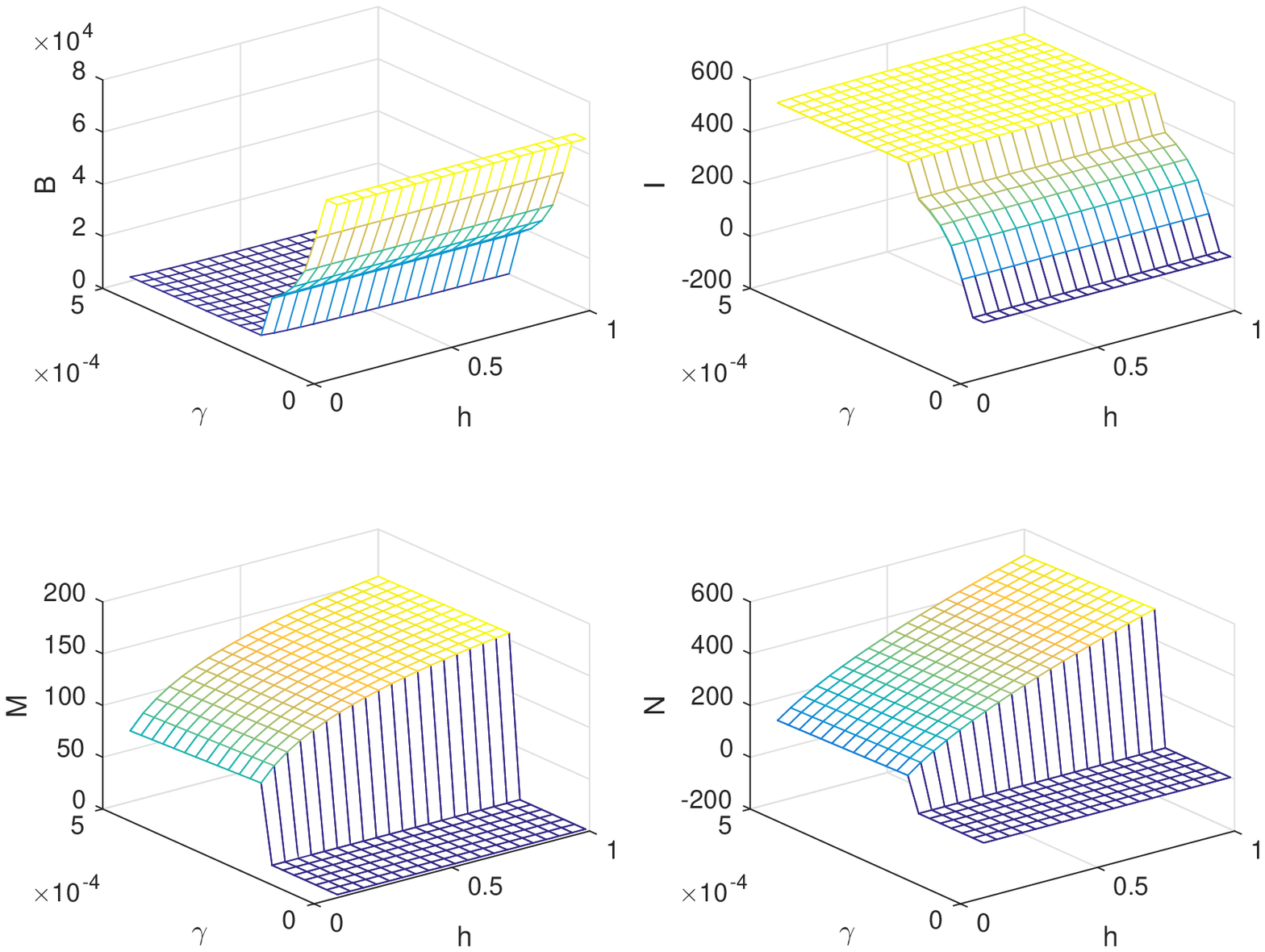}
\hfill
\includegraphics[scale=.3]{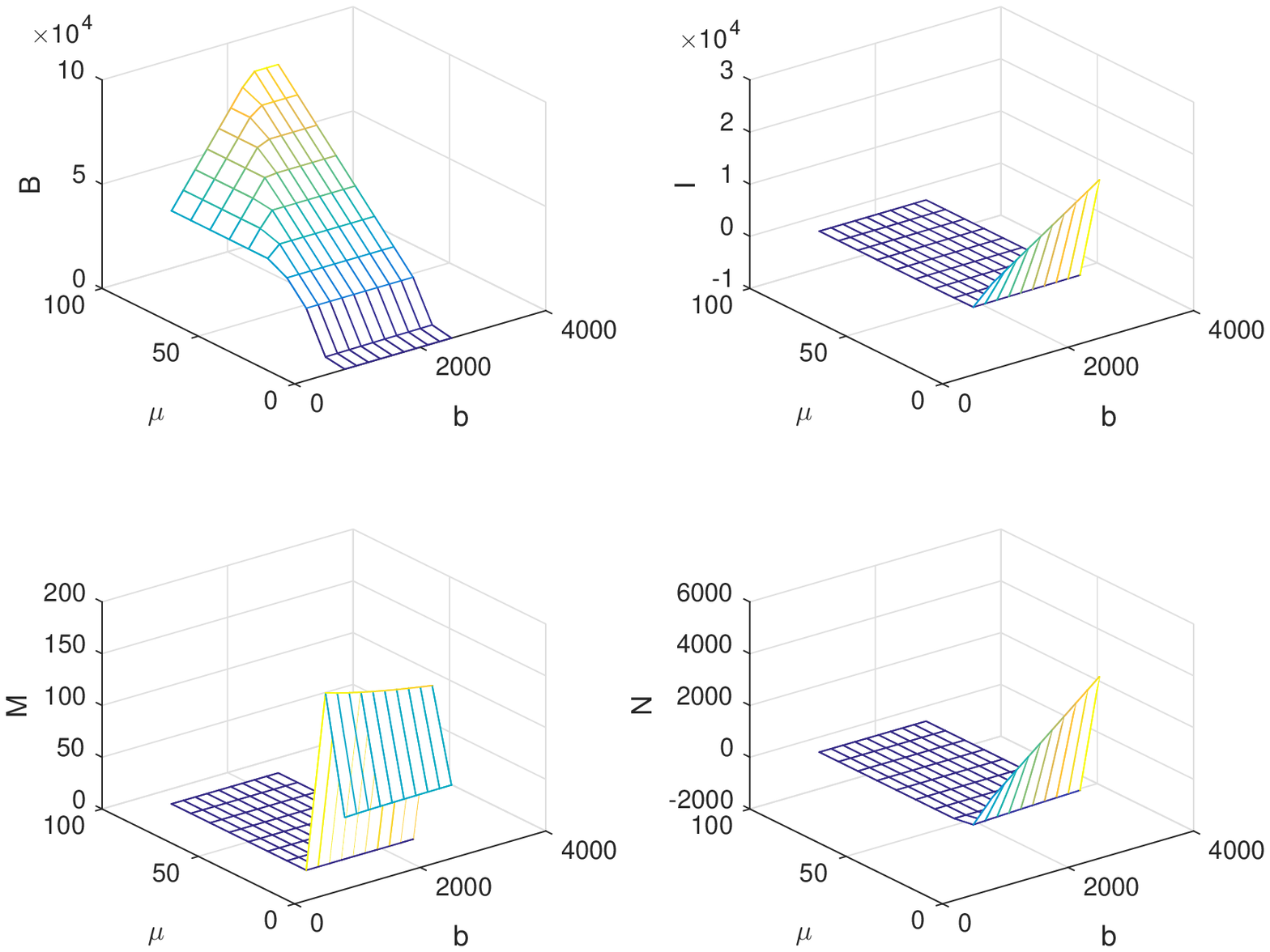}
\caption{Sensitivity surfaces, in terms of the pair of parameters $h-\gamma$, (left), and
$b-\mu$ (right). The other parameter values are choosen as $r=0.02$, $b=1500$,
$m=0.023$, $n=0.007$, $e=0.000001$, $\mu=3$, $\gamma=0.001$, $\lambda=0.004$, $\delta=0.00008$, $\beta=0.00005$, $h=0.25$, $p=0.013$.
Initial conditions $B=15000$, $I=0$, $M=6$, $N=4$.}
\label{surf_E1}
\end{figure}

In Figure \ref{r-b,r-e} (left) as the \textit{Varroa} growth rate $r$ increases,
the healthy bee population decreases while the mite populations reasonably increase. The dynamics
of the infected bee population, instead, is almost insensitive to the parameter $r$, but it
grows linearly with the daily bee birth rate $b$. 
Furthermore, this kind of behavior is observed whenever the parameter $b$ is considered,
regardless of the other parameter being taken into consideration. In fact, a larger $b$
means a greater number of bees in the colony and thus proportionally also more infected.
As in the previous Figure, infected mites benefit from a higher $b$, while the
healthy ones are instead depressed.

Another transcritical bifurcation is shown in Figure \ref{r-b,r-e} (right) connecting
the mite-free equilibrium $E_4$ and coexistence.
The bigger the \textit{Varroa} growth rate, the greater
the grooming rate needed to wipe out the mites from the beehive.
If this resistance mechanism is not high enough, the mites invade the colony. 
In the $r-n$ parameter space, all the populations behave like in this $r-e$ case.

\begin{figure}[h!]
\centering 
\includegraphics[scale=.3]{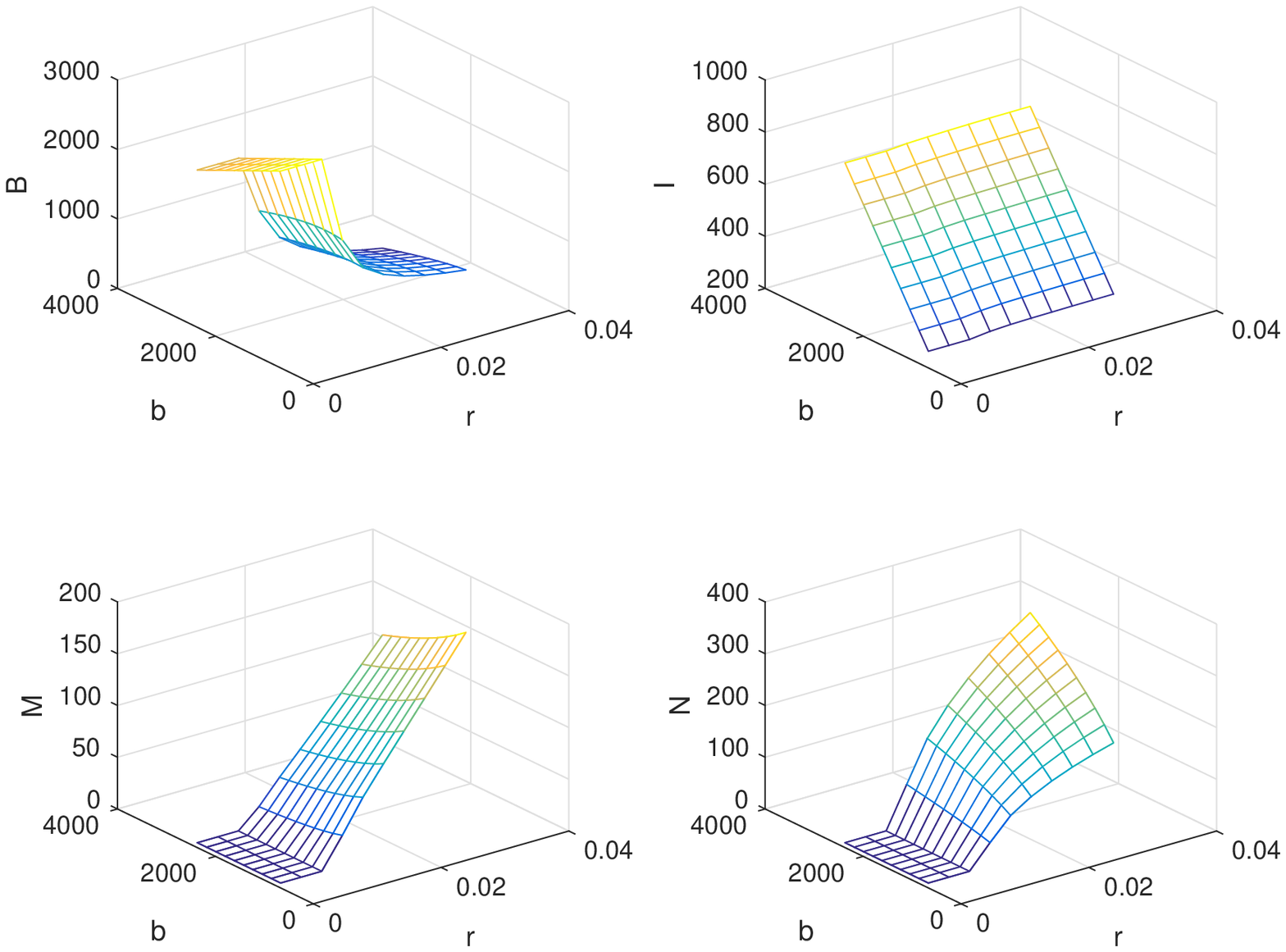}
\hfill
\includegraphics[scale=.3]{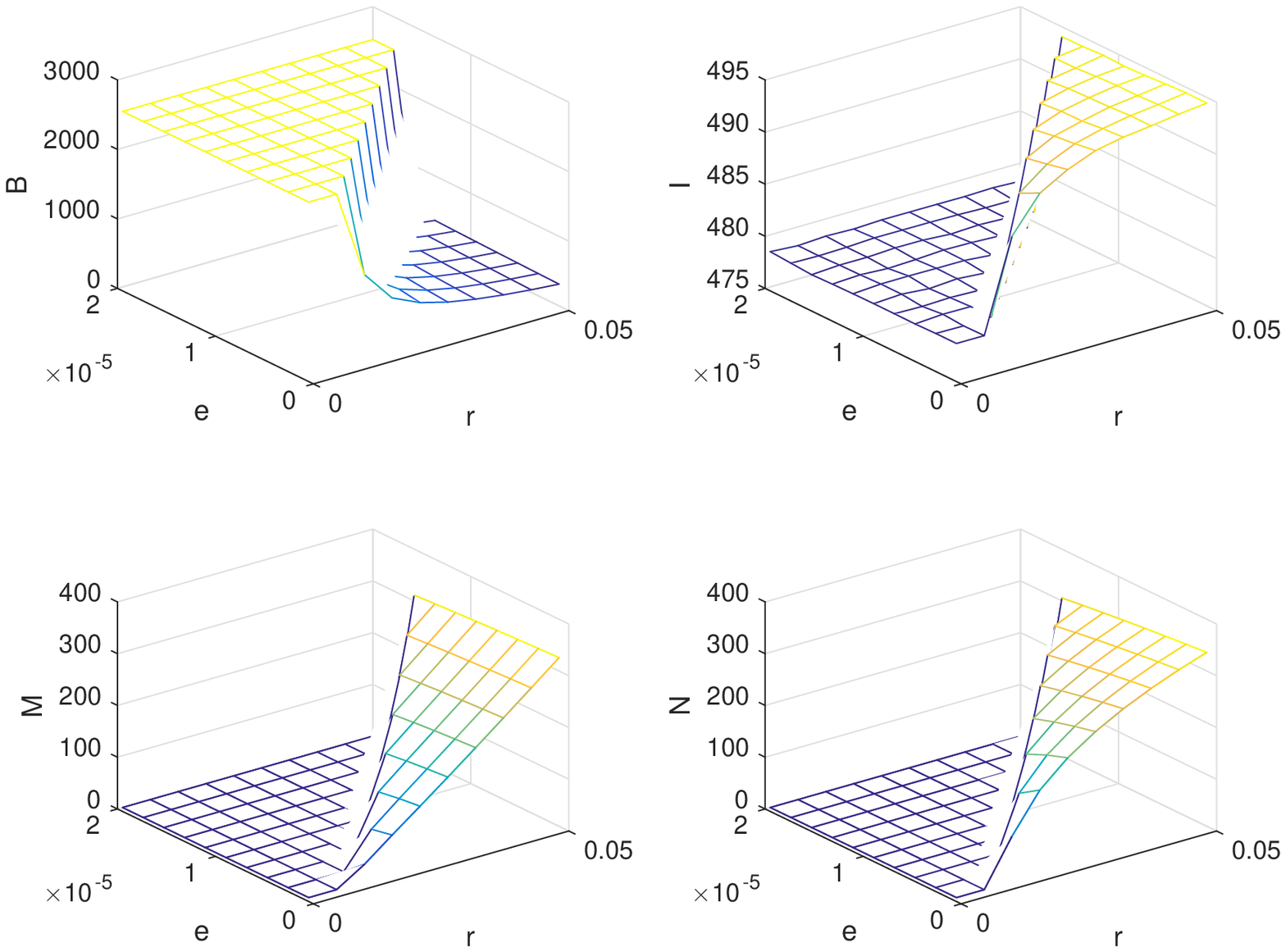}
\caption{Sensitivity surfaces, in terms of the pair of parameters $r-b$, (left), and
$r-e$ (right). The other parameter values are choosen as $r=0.02$, $b=1500$,
$m=0.023$, $n=0.007$, $e=0.000001$, $\mu=3$, $\gamma=0.001$, $\lambda=0.004$, $\delta=0.00008$, $\beta=0.00005$, $h=0.25$, $p=0.013$.
Initial conditions $B=15000$, $I=0$, $M=6$, $N=4$.}
\label{r-b,r-e}
\end{figure}

\begin{figure}[h!]
\centering 
\includegraphics[scale=.3]{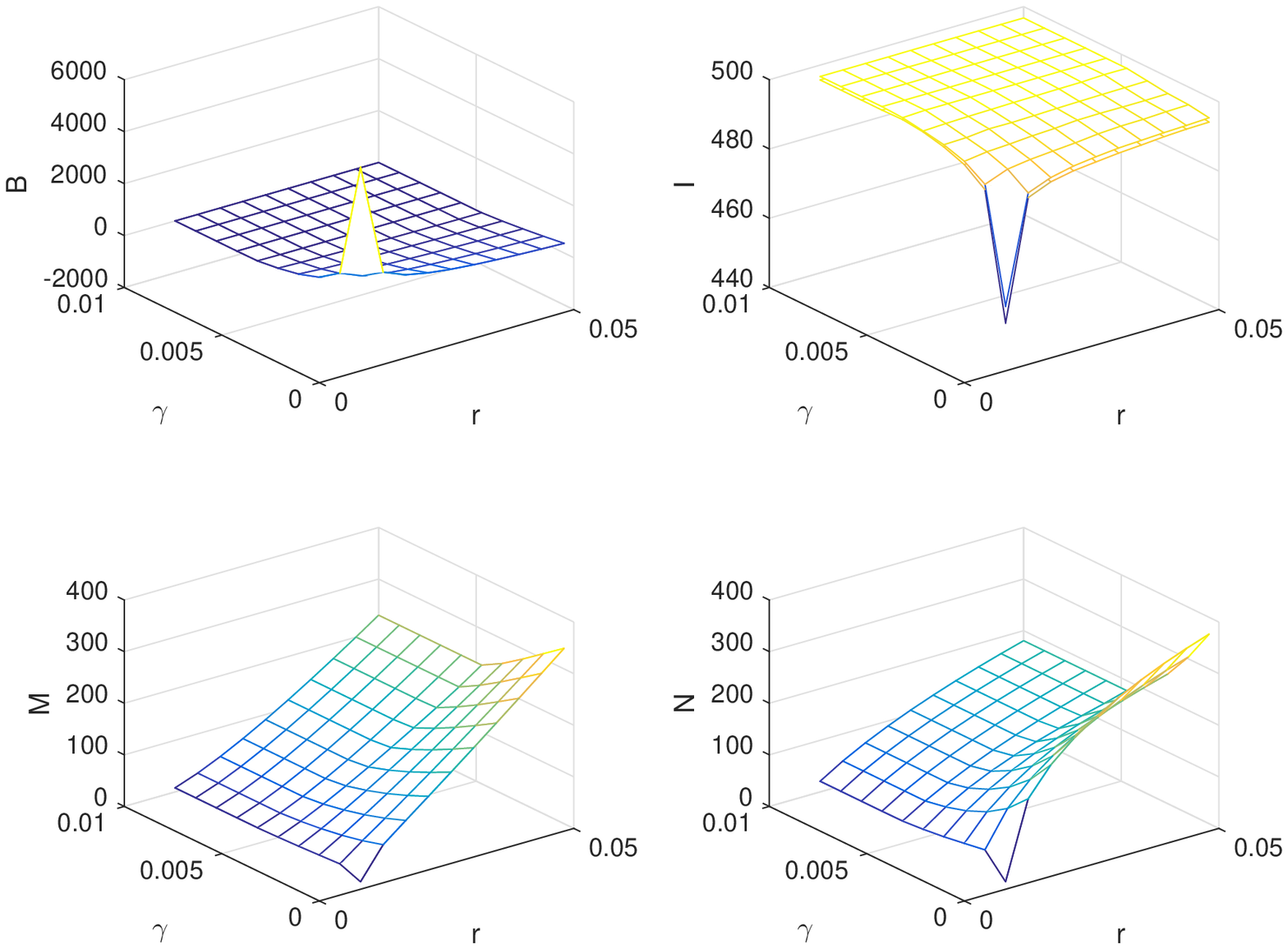}
\hfill
\includegraphics[scale=.3]{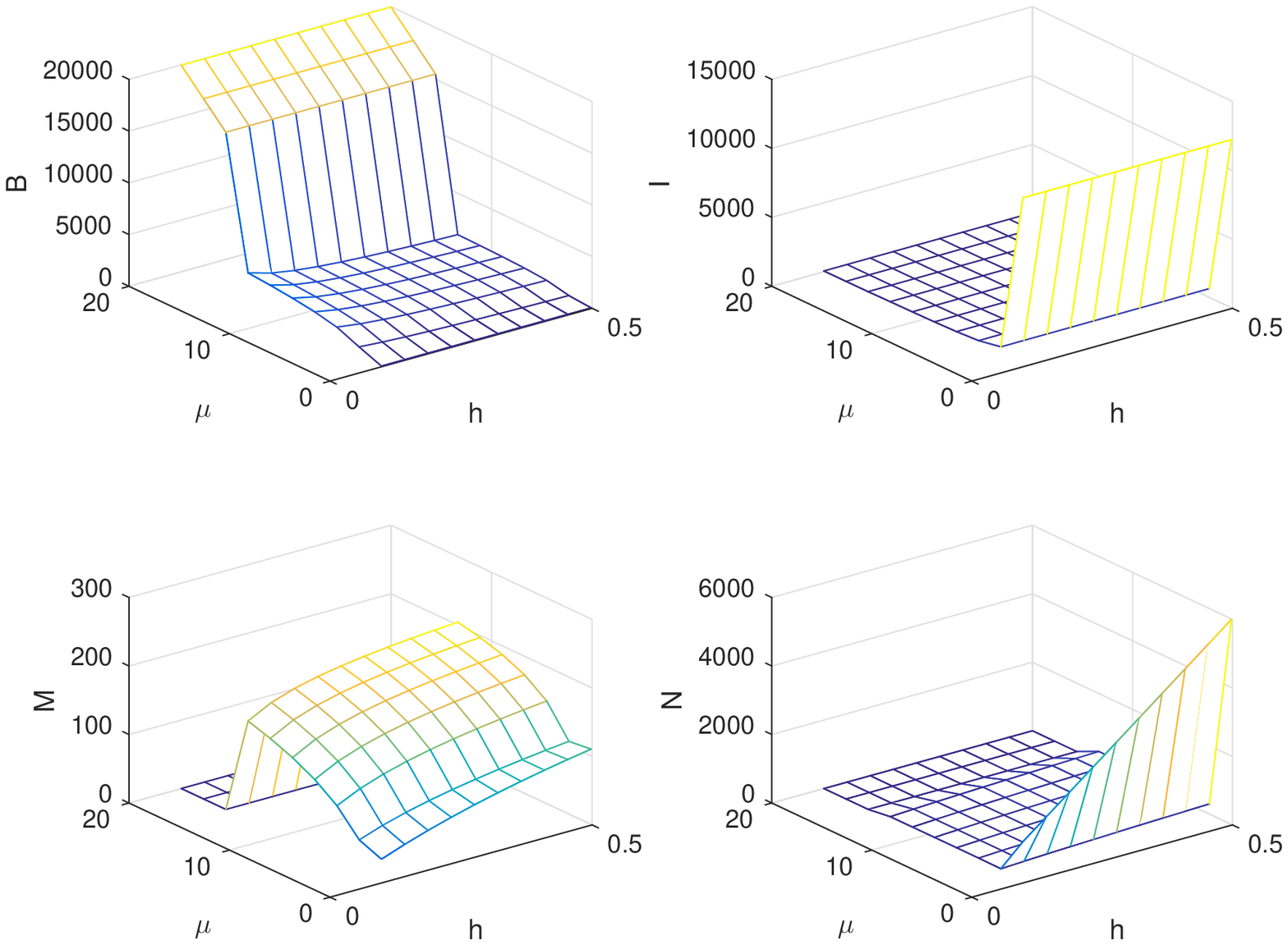}
\caption{Sensitivity surfaces, in terms of the pair of parameters $r-\gamma$, (left), and
$h-\mu$ (right). The other parameter values are choosen as $r=0.02$, $b=1500$,
$m=0.023$, $n=0.007$, $e=0.000001$, $\mu=3$, $\gamma=0.001$, $\lambda=0.004$, $\delta=0.00008$, $\beta=0.00005$, $h=0.25$, $p=0.013$.
Initial conditions $B=15000$, $I=0$, $M=6$, $N=4$.}
\label{r-gamma,h-mu}
\end{figure}

In Figure \ref{r-gamma,h-mu} (left), the combined effect of the mites growth and of the
horizontal transmission of the virus among bees are reported.
Only for really small values of both $r$ and
$\gamma$ the healthy bee population can survive.
An increase of both $\gamma$ and $r$ has a negative
influence on the healthy bee population and a positive impact on the infected bee population.
In this region of the parameter domain, the
system approaches the mite-free attractor $E_4$.
For slightly larger values of anyone of these parameters, the mites establish themselves
in the system.
As a result, we find coexistence followed, for even larger values of both such parameters,
by the healthy-bees-free equilibrium $E_5$.
This last transition occurs in particular when the bifurcation parameter $\gamma$
crosses the critical value $\gamma^{\dagger} \approx 0.005$.
Note that as $r$ increases, the mites populations resonably increase too, while they are
less sensitive to changes in $\gamma$.
The bees surfaces look also alike in the cases $\lambda-\gamma$ and $\delta-\beta$, not shown,
but the mites have a peak at the origin and are instead depressed by an increase in either
one of the parameters.

In the parameter spaces $r-\beta$, $r-\delta$ and $r-\lambda$ the bee populations behave
similarly, not shown.
A similar picture is also found in the $h-\delta$ parameter space.
A low value of $r$ independently of the other parameter leads to
the healthy-bee-only point. Then we find coexistence and ultimately
tends toward the susceptible-bee-free equilibrium.
Larger values of the parameters $\beta$ and $\delta$ increase both mites
populations but above all the infected;
instead for the latter the opposite occurs for $\lambda$.

Figure \ref{r-gamma,h-mu} (right) better shows the influence of the disease-related mortality
$\mu$ on the system dynamics.
Starting from low values of the bifurcation parameter $\mu$, we first find the
healthy-bees-free equilibrium $E_5$,
then coexistence followed by the healthy-bees-only equilibrium $E_1$.
Evidently, a higher value of the bees disease-related mortality has a positive impact on the
healthy bees and a negative influence on the infected bee population.
The highest values of the parameter ($\mu>15$) depress both mite populations.
In particular, we note that as $\mu$ increases, the populations $I$ and $N$ dramatically decrease.
Thus, a higher viral titer at the colony level is obtained for lower values of $\mu$. We thus find
again the result discussed in our previous investigation, \cite{H}:
when transmitted by Varroa mites,
the least harmful diseases for the single bees are
the most virulent ones for the whole colony.
The influence of a larger $h$ is positively felt just by the mite populations.

\begin{figure}[h!]
\centering 
\includegraphics[scale=.3]{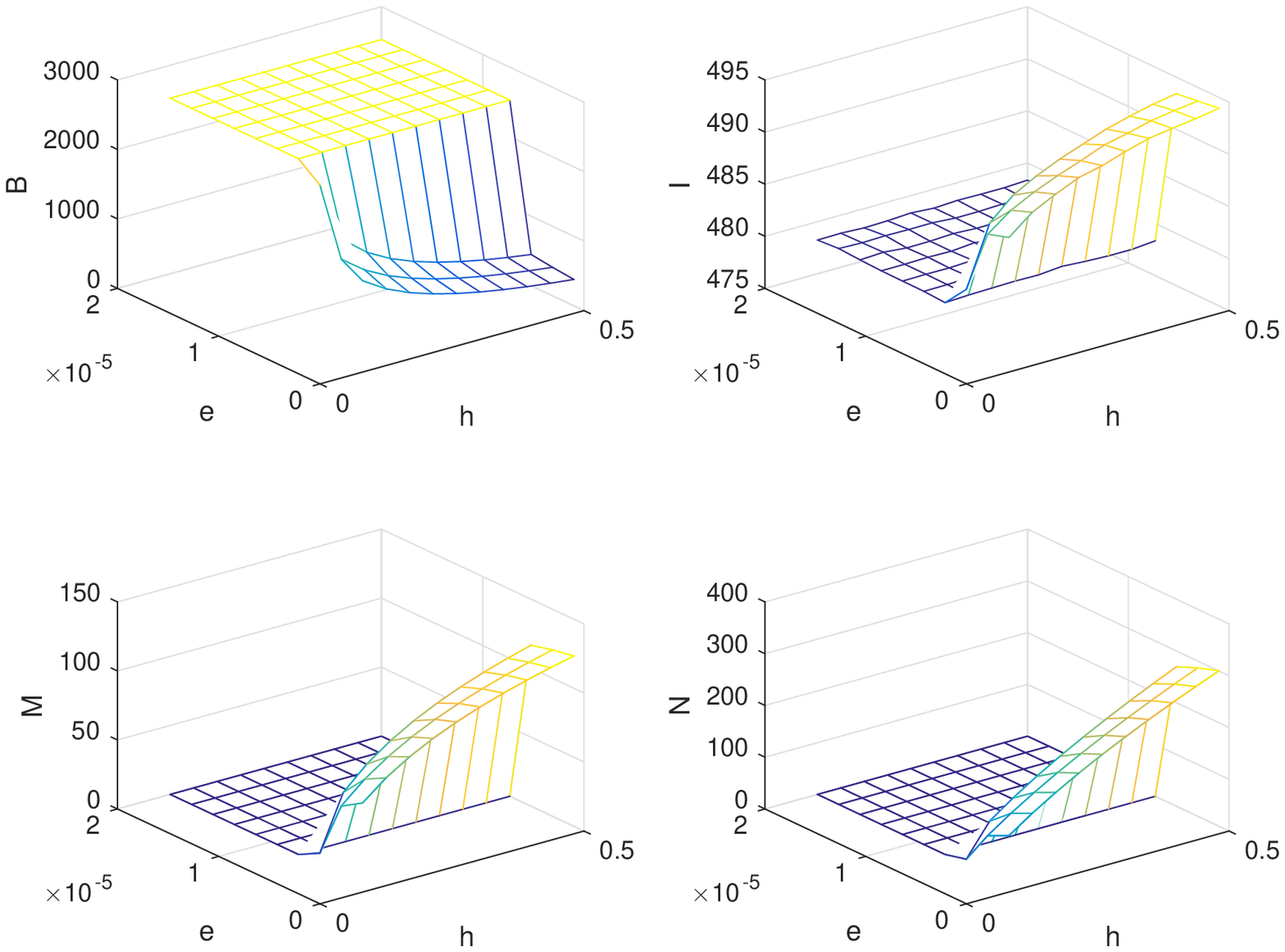}
\hfill
\includegraphics[scale=.3]{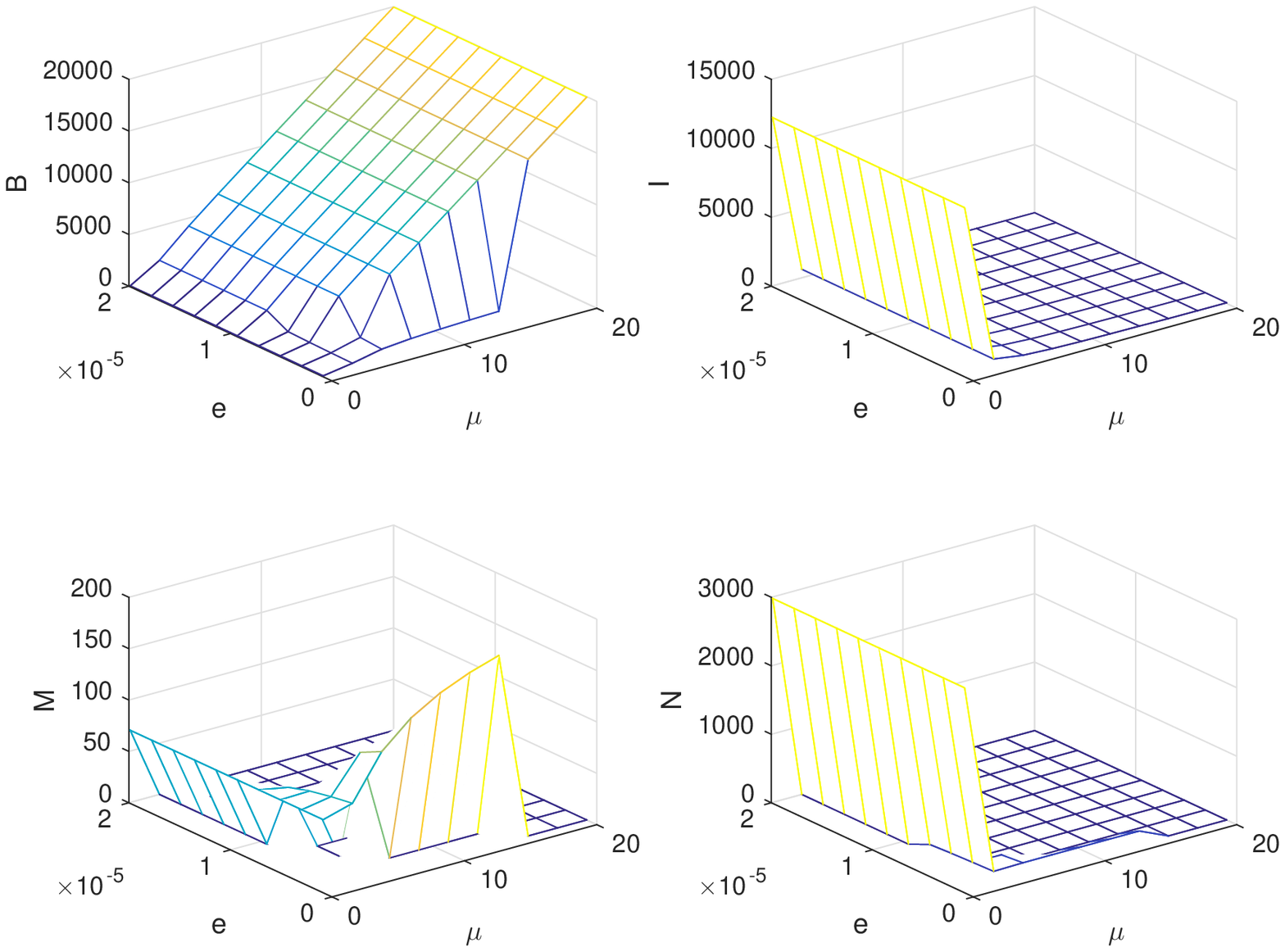}
\caption{Sensitivity surfaces, in terms of the pair of parameters $h-e$, (left), and
$e-\mu$ (right). The other parameter values are choosen as $r=0.02$, $b=1500$,
$m=0.023$, $n=0.007$, $e=0.000001$, $\mu=3$, $\gamma=0.001$, $\lambda=0.004$, $\delta=0.00008$,
$\beta=0.00005$, $h=0.25$, $p=0.013$.
Initial conditions $B=15000$, $I=0$, $M=6$, $N=4$.}
\label{h-e,mu-e}
\end{figure}

In Figure \ref{h-e,mu-e} (left) we explore the effectiveness of the grooming behavior $e$ as a
resistance mechanism against to the \textit{Varroa} infestation. The behavior of sensitivity surfaces
is encouraging: regardless of the average number of mites per bee, if the grooming behavior is
performed strongly enough mites are eventually wiped out of the system.
Indeed, a transcritical bifurcation between $E_4$ and the coexistence equilibrium occurs
when the grooming rate crosses the bifurcation threshold $e^{\dagger} \approx 0.000007$.
The higher the average number of mites per bee $h$, the larger both infected populations
become, as well as the susceptible mites, while healthy bees are depressed.
A similar behavior occurs for the case $h-n$. Also for
$m-n$ and $m-e$ we find the equilibrium surfaces to look as for the $h-e$ case,
with the only difference that in the last two cases the infected bees
benefit from a low $m$ and at least in the range explored, they do not vanish.

The bees in cases $\mu-n$ and $\mu-e$  behave similarly, see Figure \ref{e_gamma} right.
Mites and infected bees tend to disappear for larger values of $\mu$.

For $\mu-\lambda$, $\mu-\gamma$, $\mu-\delta$ and $\mu-\beta$ the infection in both
populations is hindered by larger values of the mortality, while healthy bees benefit
from it. The remaining parameter plays a role only on susceptible mites, with low
values fostering their growth.

\begin{figure}[h!]
\centering 
\includegraphics[scale=.3]{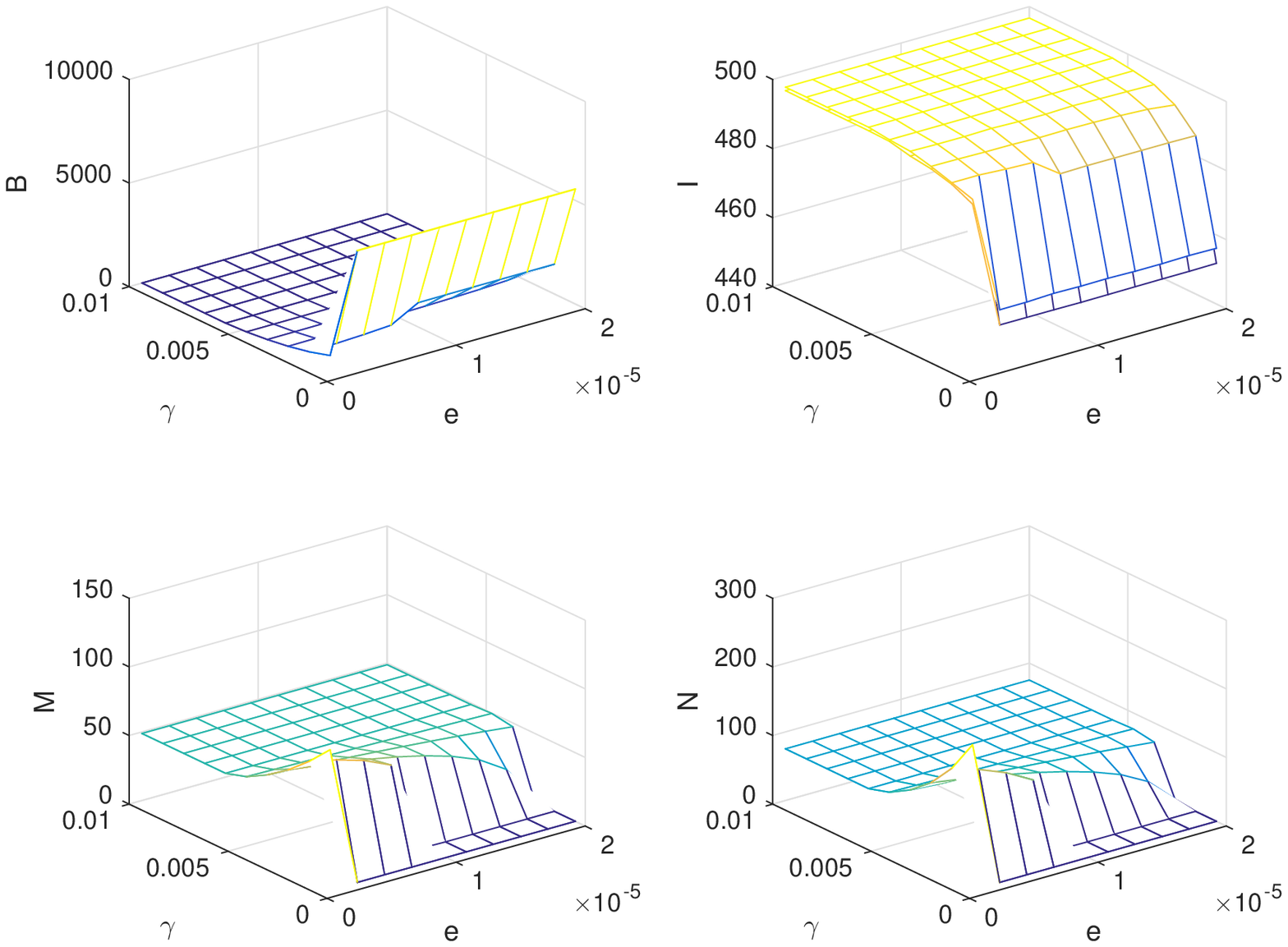}
\includegraphics[scale=.3]{bif_surf_mu_e.eps}
\caption{Sensitivity surface, in terms of the pair of parameters $e-\gamma$, (left)
and $\mu-e$ (right). 
The other parameter values are choosen as $r=0.02$, $b=1500$,
$m=0.023$, $n=0.007$, $e=0.000001$, $\mu=3$, $\gamma=0.001$, $\lambda=0.004$, $\delta=0.00008$, $\beta=0.00005$, $h=0.25$, $p=0.013$.
Initial conditions $B=15000$, $I=0$, $M=6$, $N=4$.}
\label{e_gamma}
\end{figure}

The surfaces in the parameter spaces $n-\lambda$, $n-\beta$, $e-\lambda$, $e-\delta$, $e-\beta$
and $n-\delta$ again look similar, see Figure \ref{n_delta} left. The mite-free equilibrium
is attained for relatively large values of $n$ (or $e$), low values of the remaining parameter
help only the healthy mites.

Cases $\lambda-\delta$, $\lambda-\beta$ are again similar, with larger $\lambda$ favoring
the infection in the bees and low values of the other parameter helping the healthy mites.

Again in the $\beta-\gamma$ and $\delta-\gamma$ cases the most influencial parameter
is $\gamma$, large values leading to the healthy-bee-free point. The other parameter
helps the susceptible populations if it is small.

Cases $p-h$, $p-r$ and $p-m$ are similar,
a large $p$ depresses all the populations but the healthy bees,
low values of the other parameter helping them too.
Also in the parameter spaces $p-\delta$ and $p-\beta$ the surfaces have similar shapes.
For $p-\lambda$ and $p-\gamma$ we find a similar behavior in terms of $p$, but
here the infected bees are much less affected by its growth.
For $p-b$ similar considerations
hold, but again we find that a decrease in $b$ reduces the infected bees, while $p$
on them has scant effect.

\begin{figure}[h!]
\centering 
\includegraphics[scale=.3]{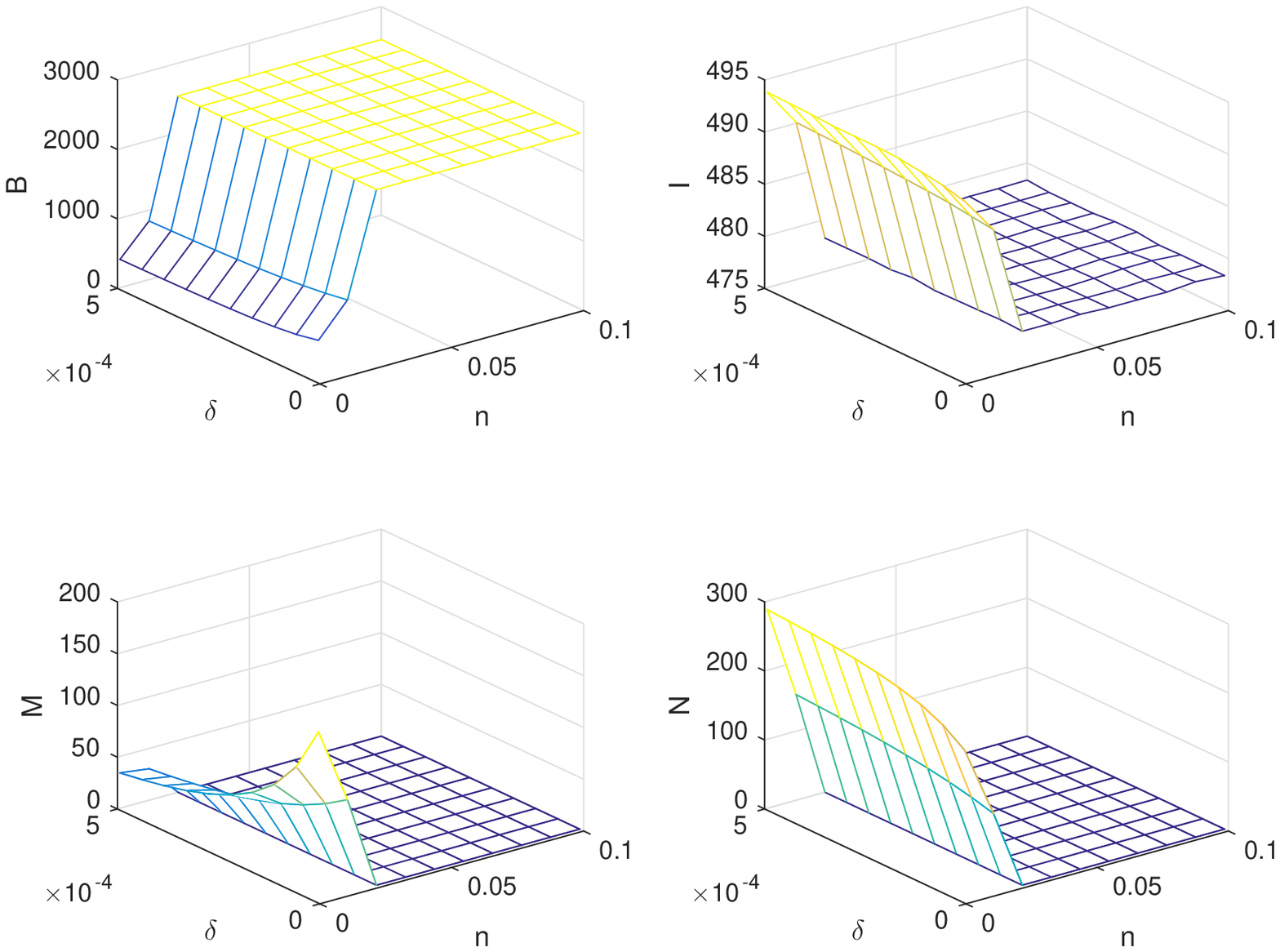}
\includegraphics[scale=.3]{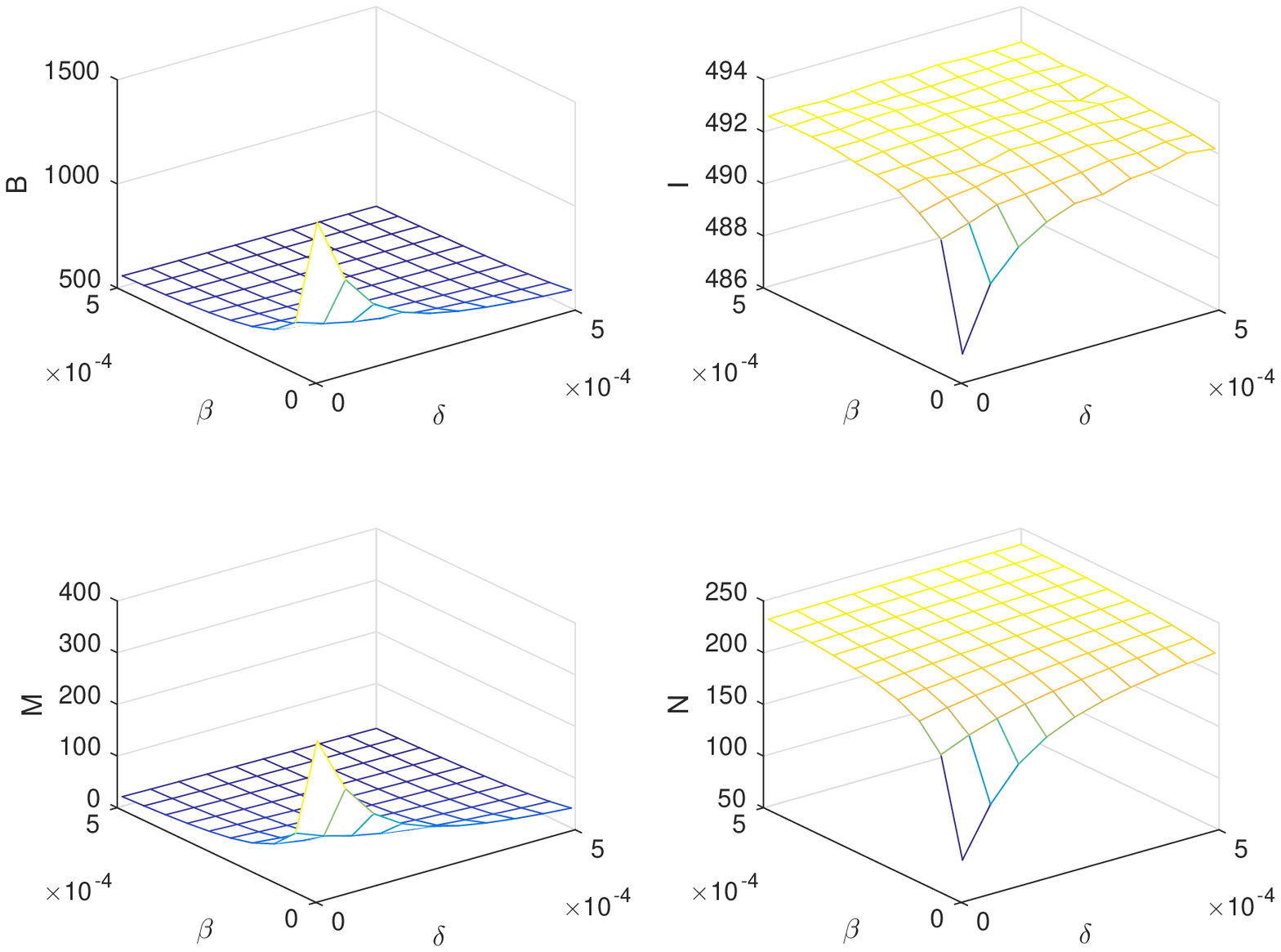}
\caption{Sensitivity surface, in terms of the pair of parameters $n-\delta$, (left)
and $\delta-\beta$ (right). 
The other parameter values are choosen as $r=0.02$, $b=1500$,
$m=0.023$, $n=0.007$, $e=0.000001$, $\mu=3$, $\gamma=0.001$, $\lambda=0.004$, $\delta=0.00008$, $\beta=0.00005$, $h=0.25$, $p=0.013$.
Initial conditions $B=15000$, $I=0$, $M=6$, $N=4$.}
\label{n_delta}
\end{figure}

\begin{figure}[h!]
\centering 
\includegraphics[scale=.3]{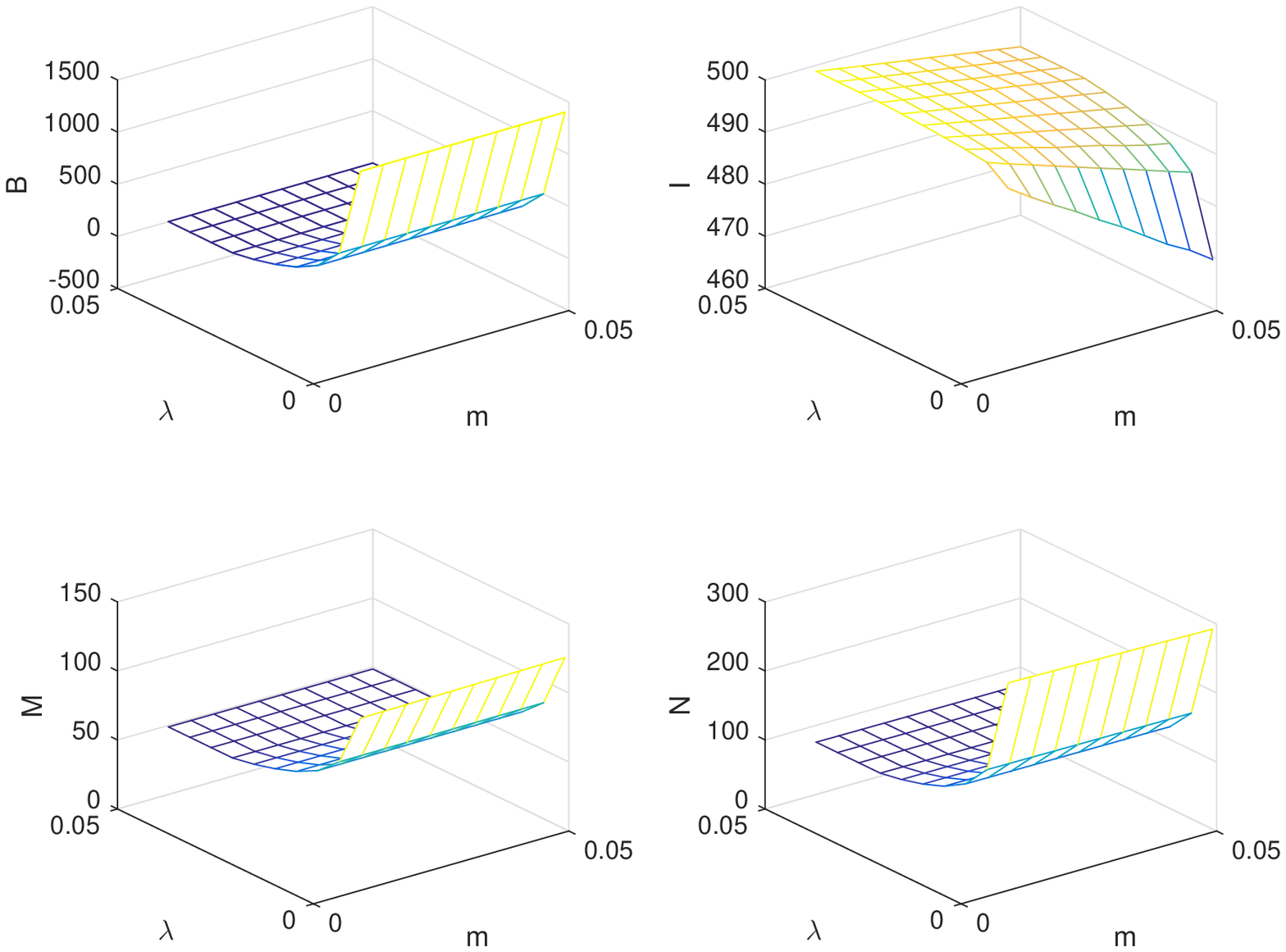}
\includegraphics[scale=.3]{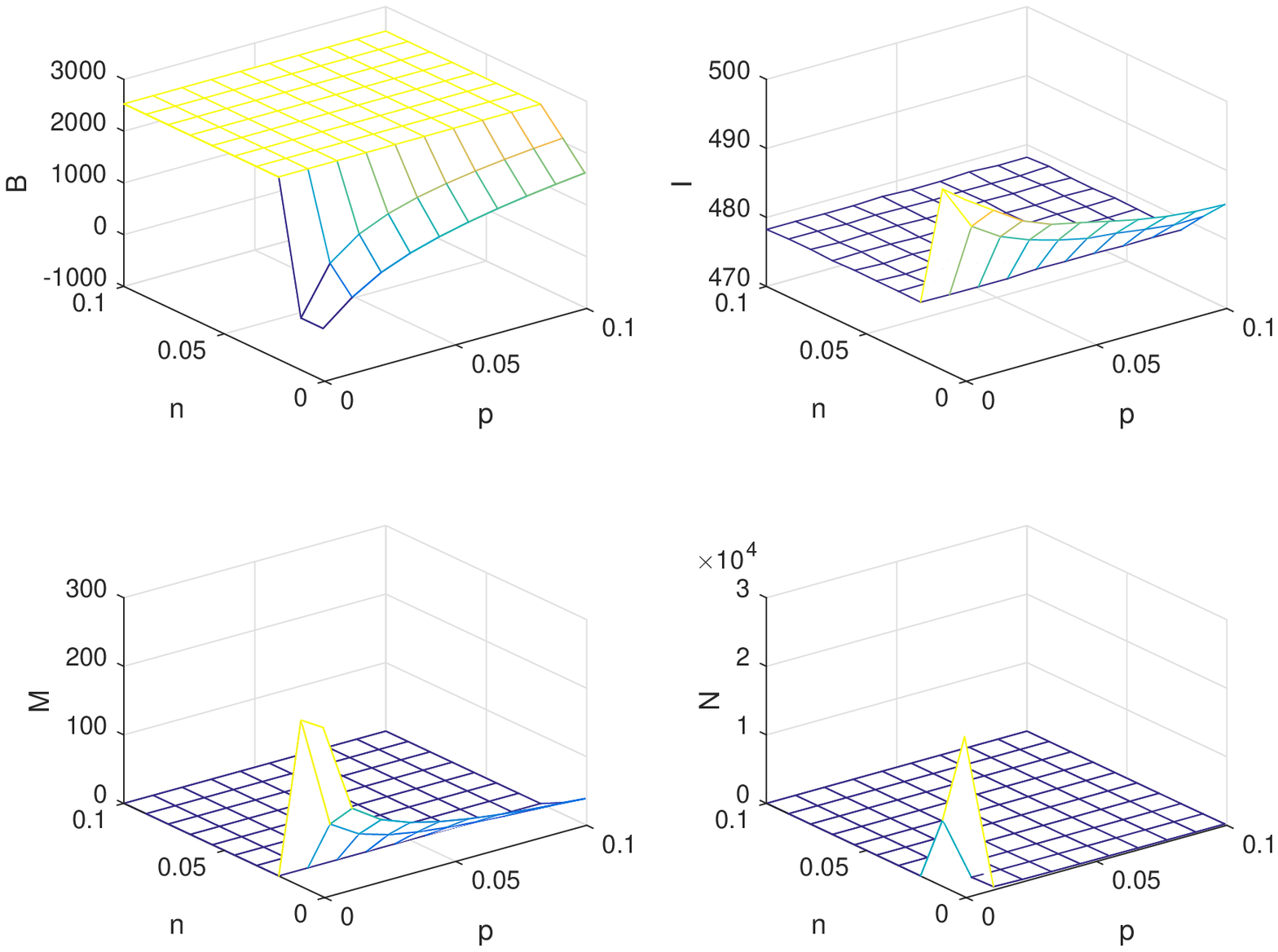}
\caption{Sensitivity surface, in terms of the pair of parameters $m-\lambda$, (left)
and $p-n$ (right). 
The other parameter values are choosen as $r=0.02$, $b=1500$,
$m=0.023$, $n=0.007$, $e=0.000001$, $\mu=3$, $\gamma=0.001$, $\lambda=0.004$, $\delta=0.00008$, $\beta=0.00005$, $h=0.25$, $p=0.013$.
Initial conditions $B=15000$, $I=0$, $M=6$, $N=4$.}
\label{m_lambda}
\end{figure}

In Figure \ref{n_delta} right,
we can observe the influence of the transmission rates $\delta$ and $\beta$.
As their values increase, both the healthy bees and mites populations decrease
while the infected ones increase.
As a result, the level of infection in the hive dramatically grows and almost
all the \textit{Varroa} mites become virus carriers.

In Figure \ref{m_lambda} left,
another transcritical bifurcation occurs. An increase of the parameter
$\lambda$ has its greatest influence on the healthy bee population and
makes it completely vanish from the system. Here the transition between
the coexistence and the no-healthy-bees equilibrium point is observed. Conversely,
the transmission rate between infected \textit{Varroa} and healthy bees
has a positive effect on the infected bees population.
Furthermore, also the mites populations decrease but evidently this is due
to the reduction of the host population.

Finally, from Figure \ref{m_lambda} right, we note the combined effect of
two parameters that hinder the infestation spreading. If the \textit{Varroa}
natural mortality is higher than the threshold value $n^{\dagger} \approx 0.023$
the mites disappear from the hive and we find equilibrium $E_4$. Instead, the
influence of \textit{Varroa} intraspecific competition is shown for smaller values
of $n$: as $p$ increases, both the mites populations and the infected bees 
decrease while the healthy bees increase. This behavior occurs also in the 
$p-e$ parameter space.

Although a sufficient rate of the grooming behavior can really help the colony
to control the \textit{Varroa} infestation,
Figure \ref{h-e,mu-e} (right) suggests that the disease-related mortality $\mu$ has a more
significant influence on the shape of the surfaces and thus
on the system dynamics than the grooming rate $e$. Indeed the shape of the
sensitivity surfaces is mostly determined by the changing parameter $\mu$.
Note here the following transcritical bifurcations: starting for low values of $\mu$ we find $E_5$,
the healthy-bee-free point,
then, as $\mu$ increases, the coexistence equilibrium and finally the mite-free equilibrium $E_4$.
This chain of transitions reemphasizes once again the effect of the parameter $\mu$.

Furthermore, from Figure \ref{e_gamma} left, the effectiveness of the grooming
behavior can be relevant only in the presence of a really low value of the parameter $\gamma$,
i.e.\ in this region where the transcritical bifurcation mentioned above between
coexistence and $E_4$ is observed, with $e$ as bifurcation parameter. Conversely, for higher values
of the horizontal transmission rate among bees, the colony is driven from $E_4$,
the mite-free situation, through
coexistence, to $E_5$ with the extinction of healthy bees.

The knowledge of these most relevant parameters could suggest the ones which
would maybe be affected by human external measures 
in order to drive the system to settle in a possibly safer position.
To sum up, the parameters most affecting the system are $r$, $\mu$, $\gamma$
and $e$.
Our findings elucidate their influence on the system.
These are however all ecosystem-related parameters, that perhaps can hardly be influenced
by man-undertaken measures, although they might depend on other external factors such as
for instance climatic changes.
Theoretically, the sensitivity
surfaces show that the population of healthy bees would highly benefit from a reduction of the
\textit{Varroa} growth rate, as well as a reduction of the horizontal transmission rate among bees.
Also a higher bees disease-related mortality and a larger grooming rate would
help in protecting the colonies.


\section{Conclusion}

We have introduced a model for bees and mites, modifying our previous approach \cite{H},
allowing a link between the carrying capacity of the mites and the bees population,
while in the former investigation these were not bound together.
Our results agree with empirical evidence, in addition to describing the fundamental
role played by the mite in this process.
The ecosystem described here can settle only to the following possible outcomes:
the disease-and-mite-free environment, the ideal situation;
the mite-free situation, in which however part of the bees are endemically infected;
the coexistence equilibrium in which both bees and mites are present, all affected
by the disease;
the \textit{Varroa} invasion leading to extinction of the healthy bees, while the remaining
ones are all infected.

Among our findings, we observe that the endemic disease cannot affect all the bees
in a \textit{Varroa}-free colony. Indeed, the infected-bees-only situation,
described by equilibrium $E_2$ in our analysis, turns out to be unstable,
if the mite reproduction rate exceeds their mortality rate, i.e.\ the opposite of
condition (\ref{assump_1}) holds.
Conversely, the \textit{Varroa} invasion scenario can become possible
where the bees are all infected and mites invade the hive, the point $E_5$ above.
In agreement with the fact that the presence of \textit{V. destructor}
increments the viral transmission,
this result indicates that the whole bee population may become infected when 
the disease vector is present in the beehive. 

Also, the two healthy populations cannot survive together, in the absence of infection.
This result agrees with field observations, in which 
the bee colony is considered infected when \textit{V. destructor} is present in the beehive.
Thus, the discovery of \textit{Varroa} in the hive necessarily implies
that at least part of the bees are virus-affected.

The findings of this study also indicate that a low horizontal virus transmission rate
among honey bees in beehives will help in protecting the bee colonies from the
\textit{Varroa} infestation and the viral epidemics. In fact, from the first
condition in \eqref{E1_stab},
a low $\gamma$ is necessary for 
the disease- and mite-free equilibrium $E_1$ to be stable.

The sensitivity analysis allows us to identify the parameters most affecting the system:
$r$, $\mu$, $\gamma$ and $e$. Specifically, 
for small changes of these parameters the system experiences transcritical bifurcations
between all the equilibria.
These results emphasize the importance of keeping the \textit{Varroa} growth under control.
In order to do this, the sensitivity surfaces suggest that a decrease of $\gamma$ combined
with an increase of $e$ will benefit the colonies.
Furthermore, the analysis substantiates the empirical remark that the most harmful diseases
at the colony level are those that least affect the single
bees. In fact, we can observe that for a very low bee mortality $\mu$, the equilibrium
toward which the system always settles is the healthy-bee-free point, in which only infected
bees and mites thrive.
The counterintuitive effect of the parameter $\mu$ on the system's dynamics agrees with the
experimental findings of \cite{Nature}. By artificially infecting the larval colonies of
{\it{Apis mellifera}} and {\it{Apis cerana}}, the researchers have compared the {\it{Varroa}}
evolution in both colonies. They found that in the european bees
it does not change much in comparison
with the colonies that are not infected, while the oriental bees were hindered, with a high
larval mortality. They conclude that the higher vulnerability of the latter to the mite could
imply a higher resistance at the colony level, because the infected larvae are more easily
spotted and killed by the worker bees. This also agrees with the general consensus that the
{\it{Varroa}} represents the most important factor affecting the survival of {\it{Apis mellifera}}
colonies, while the oriental bees thrive relatively easily in its presence.

In a similar way, a higher bees mortality induced by the diseases vectored by {\it{Varroa}}
maintains the total viral load in the colonies at a relatively moderate level. This thus
enhances their survival chances, because the individual infected bees have less time available
for horizontally transmitting the virus thereby infecting other individuals.
Both \cite{Nature} and this investigation show that vulnerable individuals may help the 
superorganism, against the common assumption that it is just the presence of ``strong''
individuals to ensure the colony survival.

Comparing the results of this investigation with the outcomes of the model \cite{H},
it appears that the present model, more mathematically elaborated in that it contains
a term of Leslie-Gower type, namely the mite bee-dependent carrying capacity, is not
really fundametally necessary, as its qualitative analysis is captured
well by the former system \cite{H}.

The bees are the most important pollinators worldwide. In recent years managed honey beehives
have been subject to decline mainly due to infections caused by the invading parasite
\textit{Varroa destructor}, leading to widspread colony collapse, \cite{GA}.
In turn, this has caused alarm
among the scientists for apiculture and even more for agriculture.
This model represents a further step beyond the basic model presented in \cite{H}
for the understanding of this problem.
But there is the need of the model validation using real data. At present, we are
currently working on
this aspect of the research, gathering the field data.
From the collected field information we plan to
obtain an estimate of the model parameters of most interest for the beekeepers: first of all the
\textit{Varroa} growth rate, a parameter most needed
to assess the amount of acaricides to be used and the timings for their most
effective application.
In order to measure real field data we have started field experiments, in collaboration with
beekeepers of the Cooperative ``Dalla Stessa Parte'', Turin, and from ``Aspromiele'', the
Association of Piedmont Honey Producers, Turin. 
These experiments represent probably the first collaboration between beekeepers
and mathematicians on this
topic in our country and they are really precious for the research,
since real field data are not yet presently available.


\begin{thebibliography}{0}

\bibitem{maci}
M. A. Benavente, R. R. Deza, M. Eguaras, Assessment of Strategies for the Control of the
\textit{Varroa destructor} mite in \textit{Apis mellifera} colonies, through a simple model,
MACI, II Congreso de matem\'atica aplicada, computacional e industrial, Editors E. M. Mancinelli,
E. A. Santill\'an-Marcus, D. A. Tarzia, Rosario, Argentina, December 14th-16th 2009, 5-8.

\bibitem{H} S. Bernardi, E. Venturino {\it Viral epidemiology of the adult Apis
Mellifera infested by the Varroa destructor mite}, HELIYON (2016).

\bibitem{api2}
E. Carpana, M. Lodesani, Editors. Patologia e avversit\`a dell'alveare
(Pathologies and adversity of beehives), springer, 2014.

\bibitem{joyce}
J. S. Figueir\'o, F. C. Coelho, The role of resistance behaviors in the
population dynamics of Honey Bees infested by \textit{Varroa Destructor},
Abstract Collection, Models in Population Dynamics end Ecology, Ezio Venturino (Editor), International Conference,
Universit\`a di Torino, Italy, August 25th-29th 2014, 23.

\bibitem{GA}
E. Genersch, M. Aubert, Emerging and re-emerging viruses of the honey bee ({\it{Apis
mellifera L.}}), Veterinary Research 41, (2010), 54.

\bibitem{izslt}
M. Milito,
Biologia della api,
IZSLT
http://www.izslt.it/apicoltura/wp-content/uploads/2012/07/BIOLOGIA-DELLE-API.pdf

\bibitem{Nature}
P. Page, Z. Lin, N. Buawangpong, H. Zheng, F. Hu, P. Neumann, P. Chantawannakul, V. Dietemann,
Social apoptosis in honey bee superorganisms,
Scientific Reports 6, Article number: 27210 
(2016)
doi:10.1038/srep27210

\bibitem{Ratti}
V. Ratti, P. G. Kevan, H. J. Eberl, A mathematical model for population dynamics in
honeybee colonies infested with {\it{Varroa destructor}} and the
{\it{Acute Bee Paralysis Virus}},
Can. Appl. Math. Q. 21 (2013), no. 1, 63-93.
DOI 10.1007/s11538-015-0093-5

\bibitem{ratti2}
V. Ratti, P. G. Kevan, H. J. Eberl,
A Mathematical Model of the Honeybee---{\it{Varroa
destructor}}---Acute Bee Paralysis Virus System with
Seasonal Effects,
Bull Math Biol
DOI 10.1007/s11538-015-0093-5

\bibitem{FCB}
J.F. Santos, F.C. Coelho, P.J. Bliman,
Behavioral modulation of the coexistence between {\it{Apis mellifera}} and
{\it{Varroa destructor}}: A defense against colony collapse? PeerJ
PrePrints 3:e1739, (2015).
doi.org/10.7287/peerj.preprints.1396v1

\bibitem{Intro} United Nations Environment Programme, Global honey bee colony disorders and other
threats to insect pollinators, UNEP Emerging Issues, (UNEP, Nairobi) (2010).

\end{thebibliography}
\end{document}